\documentclass{article}

\usepackage[preprint]{neurips_2021}

\usepackage[utf8]{inputenc} 
\usepackage[T1]{fontenc}    
\usepackage{hyperref}       
\usepackage{url}            
\usepackage{booktabs}       
\usepackage{amsfonts}       
\usepackage{nicefrac}       
\usepackage{microtype}      
\usepackage{xcolor}         

\usepackage{amsmath,amssymb,amsthm}
\usepackage{graphicx}

\newtheorem{thm}{Theorem}
\newtheorem*{thm*}{Theorem}
\newtheorem{lem}{Lemma}
\newtheorem{cor}{Corollary}

\newcommand{\EE}{ \mathbb{E} } 
\newcommand{\R}{ \mathbb{R} } 
\newcommand{\SSS}{ \mathbb{S} } 

\newcommand{\EEb}[1]{ \EE \left[ {#1} \right]}
\newcommand{\EEd}[2]{ \EEb{ \| {#1} - {#2} \|^2} }
\newcommand{\tr}[1]{ \mathrm{Tr}\left\{ {#1} \right\} }
\newcommand{\Ker}{\mathrm{Ker}}
\makeatother

\title{A Theory of the Distortion-Perception Tradeoff in Wasserstein Space}

\author{%
  Dror Freirich \\
  \texttt{drorfrc@gmail.com} \\
  \And
  Tomer Michaeli \\
  \texttt{tomer.m@ee.technion.ac.il} \\
  \AND
  Ron Meir \\
  \texttt{rmeir@ee.technion.ac.il} \\
  The Andrew and Erna Viterbi Faculty of Electrical \& Computer Engineering, \\
  Technion \\
}

\begin{document}
\setcitestyle{numbers}

\maketitle

\begin{abstract}
The lower the distortion of an estimator, the more the distribution of its outputs generally deviates from the distribution of the signals it attempts to estimate. This phenomenon, known as the perception-distortion tradeoff, has captured significant attention in image restoration, where it implies that fidelity to ground truth images comes at the expense of perceptual quality (deviation from statistics of natural images). However, despite the increasing popularity of performing comparisons on the perception-distortion plane, there remains an important open question: \emph{what is the minimal distortion that can be achieved under a given perception constraint?} In this paper, we derive a closed form expression for this distortion-perception (DP) function for the mean squared-error (MSE) distortion and the Wasserstein-2 perception index. We prove that the DP function is always quadratic, regardless of the underlying distribution. This stems from the fact that estimators on the DP curve form a geodesic in Wasserstein space. In the Gaussian setting, we further provide a closed form expression for such estimators. For general distributions, we show how these estimators can be constructed from the estimators at the two extremes of the tradeoff: The global MSE minimizer, and a  minimizer of the MSE under a perfect perceptual quality constraint. The latter can be obtained as a stochastic transformation of the former.

\end{abstract}

\section{Introduction}
	
		Image restoration covers some fundamental settings in image processing such as denoising, deblurring and super-resolution. Over the past few years, image restoration methods have demonstrated impressive improvements in both visual quality and distortion measures such as peak signal-to-noise ratio (PSNR) and structural similarity index (SSIM) \citep{wang2004image}. It was noticed, however, that improvement in accuracy, as measured by distortion, does not necessarily lead to improvement in visual
		quality, referred to as perceptual quality. Furthermore, the lower the distortion of an estimator, the more the distribution of its outputs generally deviates from the distribution of the signals it attempts to estimate. 
		This phenomenon, known as the \emph{perception-distortion tradeoff} \citep{blau2018perception}, has captured significant attention, where it implies that faithfulness to ground truth images comes at the expense of perceptual quality, namely the deviation from statistics of natural images. Several works have extended the perception-distortion tradeoff to settings such as lossy compression \citep{blau2019rethinking} and classification \citep{liu2019classification}.
		
	Despite the increasing popularity of performing comparisons on the perception-distortion plane, the exact characterization of the minimal distortion that can be achieved under a given perception constraint remains an important open question.
	Although \citet{blau2018perception} investigated the basic properties of this \emph{distortion-perception function}, such as monotonicity and convexity, little is known about its precise nature. While a general answer to this question is unavailable, in this paper, we derive a closed form expression for the distortion-perception (DP) function for the mean squared-error (MSE) distortion and the Wasserstein-$2$ perception index. 
		
	Our main contributions are:  
		\emph{(i)} We prove that the DP function is \emph{always} quadratic in the perception constraint $P$,
		regardless of the underlying distribution (Theorem \ref{Thm:=00005Bthe-Distortion-Perception-funct}).
			\emph{(ii)} We show that it is possible to construct estimators on the DP curve from the estimators at the two extremes of the tradeoff (Theorem~\ref{Thm::extrapol}): The one that globally minimizes the MSE, and a minimizer of the MSE under a perfect perceptual quality constraint. The latter can be obtained as a stochastic transformation of the former.
			\emph{(iii)} In the Gaussian setting, we further provide a closed form expression for optimal estimators and for the corresponding DP curve (Theorems \ref{thm:Gaussian1} and \ref{thm:Gaussian_not_unique}).
			We show this Gaussian DP curve is a lower bound on the DP curve of any distribution having the same second order statistics. Finally, we illustrate our results, numerically and visually, in a super-resolution setting in Section \ref{sec::numerical}. The proofs of all the theorems in the main text are provided in Appendix \ref{APPsec::proofs}.

Our theoretical results shed light on several topics that are subject to much practical activity. Particularly, in the domain of image restoration, numerous works target perceptual quality rather than distortion (\emph{e.g.} \citep{wang2018esrgan,lim2017enhanced,ledig2017photo}). However, it has recently been recognized that generating a single reconstructed image often does not convey to the user the inherent ambiguity in the problem. Therefore, many recent works target \emph{diverse} perceptual image reconstruction, by employing randomization among possible restorations \citep{lugmayr2020srflow,bahat2020explorable,prakash2021removing,abid2021generative}. Commonly, such works perform sampling from the posterior distribution of natural images given the degraded input image. This is done \emph{e.g.}~using priors over image patches \citep{friedman2021posterior}, conditional generative models \citep{ohayon2021high,prakash2020divnoising}, or implicit priors induced by deep denoiser networks \citep{kawar2021stochastic}. Theoretically, posterior sampling leads to perfect perceptual quality (the restored outputs are distributed like the prior). However, a fundamental question is whether this is optimal in terms of distortion. As we show in Section~\ref{sec::The MSE--Wasserstein-2 tradeoff}, posterior sampling is often not an optimal strategy, in the sense that there often exist perfect perceptual quality estimators that achieve lower distortion.

Another topic of practical interest, is the ability to \emph{traverse the distortion-perception tradeoff} at test time, without having to train a different model for each working point. Recently, interpolation has been suggested for controlling several objectives at test-time. \citet{shoshan2019dynamic} propose using interpolation in some latent space in order to approximate intermediate objectives. \citet{wang2018esrgan} use per-pixel interpolation for balancing perceptual quality and fidelity. Studies of network parameter interpolation are presented by \citet{wang2018esrgan,wang2019deep}.  \citet{deng2018enhancing} produces a low distortion reconstruction and a high perceptual quality one, and then uses style transfer to combine them. An important question, therefore, is which strategy is optimal. In Section~\ref{sec::optimalestimators} we show that for the MSE--Wasserstein-2  tradeoff, linear interpolation leads to optimal estimators. We also discuss a geometric connection between interpolation and the fact that estimators on the DP curve form a geodesic in Wasserstein space.

	\section{Problem setting and preliminaries}
	\label{Sec::D-P::optimalTransport}

	\subsection{The distortion-perception tradeoff}
	Let $X,Y$ be random vectors taking values in $\R^{n_{x}}$ and $\R^{n_{y}}$, respectively. We consider the problem of constructing an estimator $\hat{X}$ of $X$ based on $Y$. Namely, we are interested in determining a conditional distribution $p_{\hat{X}|Y}$ such that $\hat{X}$ constitutes a good estimate of $X$.
	
	In many practical cases, the goodness of an estimator is associated with two factors: (i) the degree to which $\hat{X}$ is close to $X$ on average (low distortion), and (ii) the degree to which the distribution of $\hat{X}$ is close to that of $X$ (good perceptual quality). An important question, then, is \emph{what is the minimal distortion that can be achieved under a given level of perceptual quality?} and \emph{how can we construct estimators that achieve this lower bound?}
	In mathematical language, we are interested in analyzing the distortion-perception (DP) function (defined similarly to the perception-distortion function of \cite{blau2018perception})
	\begin{equation}
		D(P)=\min_{p_{\hat{X}|Y}}\left\{\EE[d(X,\hat{X})] \;:\; d_{p}(p_X,p_{\hat{X}})\leq P\right\}.\label{eq:D_P::General_definition}
	\end{equation}
	Here, $d:\R^{n_x}\times\R^{n_x}\rightarrow \R^{+}\cup\{0\}$ is some distortion criterion, $d_{p}(\cdot,\cdot)$ is some divergence between probability measures, and $p_{\hat{X}}$ is the probability measure on $\R^{n_x}$ induced by $p_{\hat{X}|Y}$ and $p_Y$. We assume that $\hat{X}$ is independent of $X$ given $Y$.
	
	As discussed in \citep{blau2018perception}, the function $D(P)$ is monotonically non-increasing and is convex whenever $d_{p}(\cdot,\cdot)$ is convex in its second argument (which is the case for most popular divergences). However, without further concrete assumptions on the distortion measure $d(\cdot,\cdot)$ and the perception index $d_p(\cdot,\cdot)$, little can be said about the precise nature of $D(P)$.

    Here, we focus our attention on the squared-error distortion $d(x,\hat{x})=\|x-\hat{x}\|^{2}$ and the Wasserstein-2 distance $d_p(p_X,p_{\hat{X}})=W_2(p_X,p_{\hat{X}})$, with which \eqref{eq:D_P::General_definition} reads
	\begin{equation}
		D(P)=\min_{p_{\hat{X}|Y}}\left\{\EE[\|X-\hat{X}\|^2] \;:\; W_2(p_X,p_{\hat{X}})\leq P\right\}.
		\label{eq:MSE::D_P::Definition}
	\end{equation}
    Throughout this paper we assume that all distributions have finite first and second moments. In addition, from Theorem \ref{Thm::extrapol} below it will follow that the minimum is indeed attained, so that \eqref{eq:MSE::D_P::Definition} is well defined. It is well known that the estimator minimizing the mean squared error (MSE) without any constraints, is given by $X^{*}=\EE[X|Y]$. This implies that $D(P)$ monotonically decreases until $P$ reaches 
    $P^*\triangleq W_2(p_X,p_{X^*})$,
    beyond which point $D(P)$ takes the constant value $D^*\triangleq \EE[\|X-X^{*}\|^{2}]$.
     This is illustrated in Fig.~\ref{fig:DP_func}. 
     It is also known that in this case $D(0)\leq 2D^*$ since the posterior sampling estimator $p_{\hat{X}|Y}=p_{X|Y}$ achieves $W_2(p_X,p_{\hat{X}})=0$ and $\EE[\|X-\hat{X}\|^2]=2D^*$ \citep{blau2018perception}. However, apart for these rather general properties, the precise shape of the DP curve has not been determined to date, and neither have the estimators that achieve the optimum in \eqref{eq:MSE::D_P::Definition}. This is our goal in this paper.

\begin{figure}
		\centering
		\includegraphics[width=0.44\linewidth]{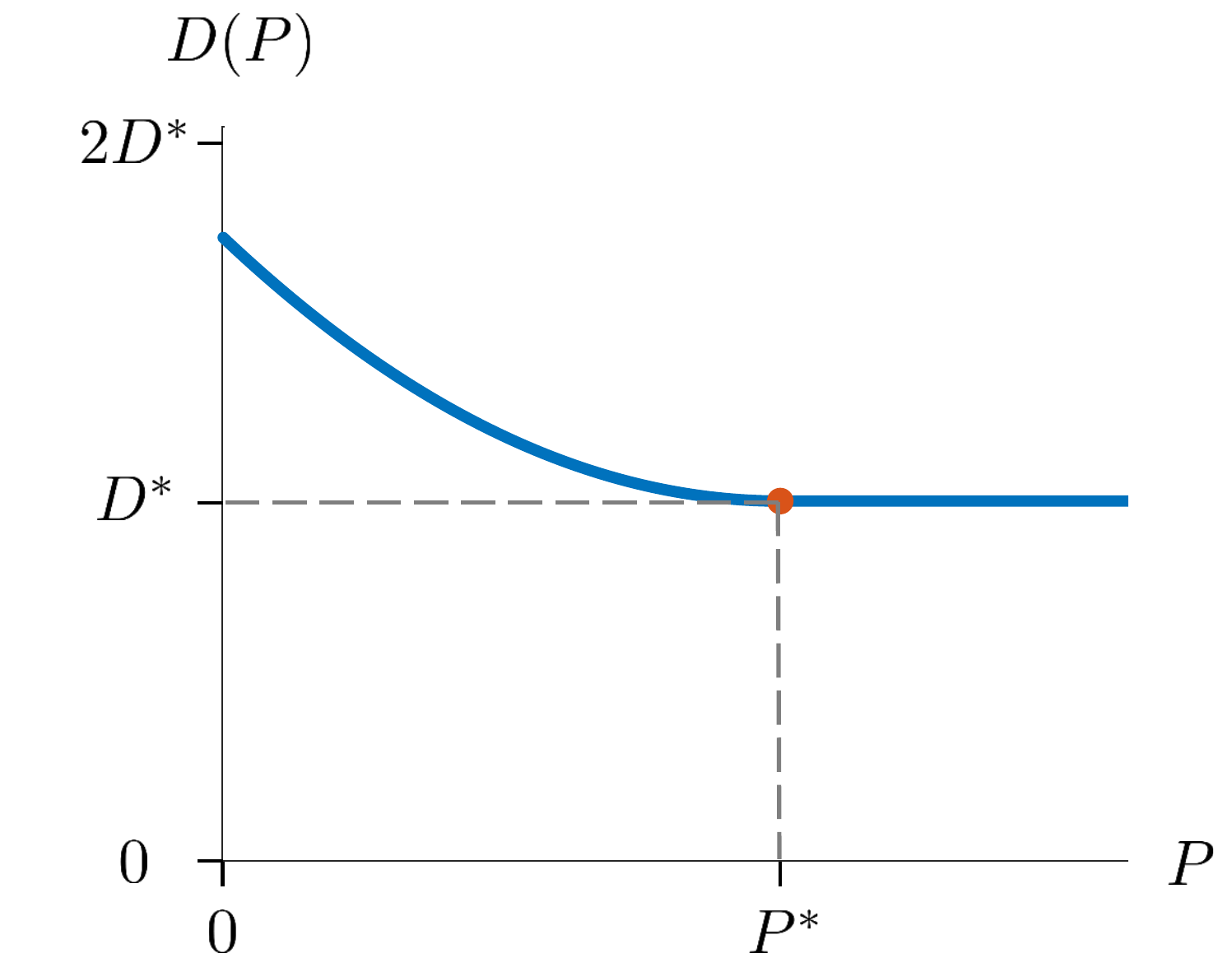} \includegraphics[bb=39bp 62bp 420bp 300bp,clip,scale=0.75,width=0.55\linewidth]{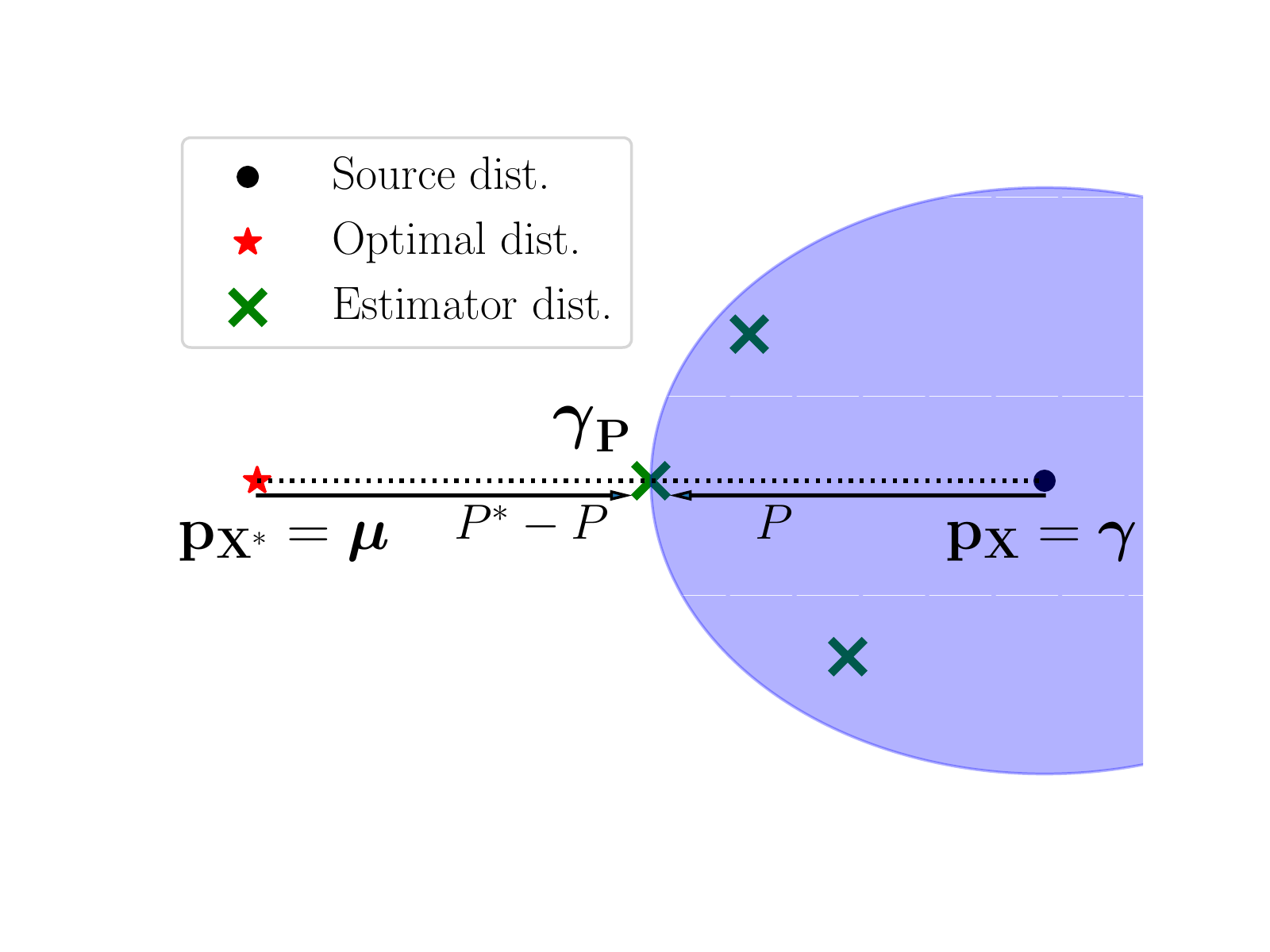}
		
		\caption{\textbf{Left: The distortion-perception function.} When using the MSE distortion and the Wasserstein-2 perception index, the minimal possible distortion, $D^*$, is achieved by the estimator $X^{*}=E[X|Y]$. The perception index attained by this estimator is $P^*$. At the other extreme of the tradeoff, we know that the distortion at $P=0$ is bounded from above by $2D^*$.\label{fig:DP_func} \textbf{Right:} \label{fig:geodesic} The minimal distortion $D(P)$ for a given perception
	index $P<P^*$ can be achieved by an estimator with a distribution $\gamma_{P}$
	lying on a straight line (or geodesic) defined by the geometry of
	the probabilities space.
	Given $P$, $\gamma_{P}$ achieves $W_{2}(p_X,\gamma_{P})=P$ and $W_{2}(p_{X^*},\gamma_{P})=P^*-P$,
	hence $D(P)=D^*+W_2^2(p_{X^*},\gamma_{P})=D^*+(P^*-P)^{2}.$}
	\end{figure}
	
	\subsection{The Wasserstein and Gelbrich Distances}
	
	Before we present our main results, we briefly survey a few properties of the Wasserstein distance, mostly taken from \citep{panaretos2020invitation}. The Wasserstein-$p$ ($p\geq 1$) distance between measures $\mu$ and $\gamma$ on a separable Banach space $\mathcal{X}$ with norm $\| \cdot \|$ is defined by
	\begin{equation}\label{eq:Wp_def}
	W_{p}^{p}(\mu,\gamma)\triangleq\inf\left\{ \EE_{(U,V)\sim\nu}[\|U-V\|^{p}]\;:\;\nu\in\Pi(\mu,\gamma)\right\},
	\end{equation}
	where $\Pi(\mu,\gamma)$ is the set of all probabilities on $\mathcal{X}\times\mathcal{X}$ with marginals $\mu$ and $\gamma$. A joint probability $\nu$ achieving the optimum in \eqref{eq:Wp_def} is often referred to as \emph{optimal plan}. The Wasserstein space of probability measures is defined as
	\[
	\mathcal{W}_{p}(\mathcal{X})\triangleq\left\{\gamma:\int_{\mathcal{X}}\|x\|^{p}d\gamma<\infty\right\},
	\]
	and $W_p$ constitutes a metric on $\mathcal{W}_{p}(\mathcal{X})$.

    For any $(m_{1},\Sigma_{1}),(m_{2},\Sigma_{2})\in \R^{d}\times \SSS_{+}^{d}$
	(where $\SSS_{+}^{d}$ is the set of symmetric positive semidefinite matrices in $\R^{d\times d})$, the Gelbrich distance is defined as
	\begin{equation}
	G^{2}((m_{1},\Sigma_{1}),(m_{2},\Sigma_{2}))\triangleq\Vert m_{1}-m_{2}\Vert_{2}^{2}+\mathrm{Tr}\left\{\Sigma_{1}+\Sigma_{2}-2\left(\Sigma_{1}^{\frac{1}{2}}\Sigma_{2}\Sigma_{1}^{\frac{1}{2}}\right)^{\frac{1}{2}}\right\}.\label{eq:Gelbrich_dist}
	\end{equation}
	The root of a PSD matrix is always taken to be PSD. For any two probability measures $\mu_{1},\mu_{2}$ on $\R^{d}$ with means and covariances $(m_{1},\Sigma_{1}),(m_{2},\Sigma_{2})$, from \citep[Thm. 2.1]{gelbrich1990formula} we have that 
	\begin{equation}\label{eq:W2_greater_G2}
	W_{2}^{2}(\mu_{1},\mu_{2})\geq G^{2}((m_{1},\Sigma_{1}),(m_{2},\Sigma_{2})).
	\end{equation}
	
	When $\mu_1=\mathcal{N}(m_{1},\Sigma_{1})$ and $\mu_2=\mathcal{N}(m_{2},\Sigma_{2})$ are Gaussian distributions on $\R^{d}$, we have that $W_{2}(\mu_{1},\mu_{2})= G((m_{1},\Sigma_{1}),(m_{2},\Sigma_{2}))$.
	This equality is obvious for non-singular measures but is true for any two Gaussian distributions \citep[p.~18]{panaretos2020invitation}.
	If $\Sigma_1$ and $\Sigma_2$ are non-singular, then the distribution attaining the optimum in \eqref{eq:Wp_def} corresponds to
	\begin{equation}\label{eq:optimal_dist_for_gauss}
	U\sim \mathcal{N}(m_1,\Sigma_1),\;\; V=m_{2}+T_{1\rightarrow 2}(U-m_{1}),
	\end{equation}
	where
	\begin{equation}\label{eq::T1->2:def}
	T_{1\rightarrow 2}=\Sigma_{1}^{-\frac{1}{2}}\left(\Sigma_{1}^{\frac{1}{2}}\Sigma_{2}\Sigma_{1}^{\frac{1}{2}}\right)^{\frac{1}{2}}\Sigma_{1}^{-\frac{1}{2}}
	\end{equation}
	is the optimal transformation pushing forward from $\mathcal{N}(0,\Sigma_{1})$ to $\mathcal{N}(0,\Sigma_{2})$ \citep{knott1984optimal}. This transformation satisfies $\Sigma_{2}=T_{1\rightarrow 2}\Sigma_{1}T_{1\rightarrow 2}.$ For a discussion on singular distributions, please see App.~\ref{Appsec::Supp}.
	
	\section{Main results}
	
	\subsection{The MSE--Wasserstein-2 tradeoff}
	\label{sec::The MSE--Wasserstein-2 tradeoff} The DP function \eqref{eq:MSE::D_P::Definition} depends, of course, on the underlying joint probability $p_{XY}$ of the signal $X$ and measurements $Y$. Our first key result is that this dependence can be expressed solely in terms of $D^*$ and $P^*$. In other words, knowing the distortion and perception index attained by the minimum MSE estimator $X^*$, suffices for determining $D(P)$ for any $P$.
	
	\begin{thm}[The DP function]
		\label{Thm:=00005Bthe-Distortion-Perception-funct} The DP function \eqref{eq:MSE::D_P::Definition} is given by
		\begin{equation}
		D(P)=
		D^*+\left[(P^*-P)_{+}\right]^{2},
		\label{eq:Thm:DP_function::Main_result}
		\end{equation}
		where $(x)_+=\max(0,x)$. Furthermore, an estimator achieving perception index $P$ and distortion $D(P)$ can always be constructed by applying a (possibly stochastic) transformation to $X^*$.
	\end{thm}
	
	Theorem~\ref{Thm:=00005Bthe-Distortion-Perception-funct} is of practical importance because in many cases constructing an estimator that achieves a low MSE (\emph{i.e.} an approximation of $X^*$) is a rather simple task. This is the case, for example, in image restoration with deep neural networks. There, it is common practice to train a network by minimizing its average squared error on a training set. Now, measuring the MSE of such a network on a large test set allows approximating $D^*$. We can also obtain an approximation of at least a lower bound on $P^*$ by estimating the second order statistics of $X$ and $X^*$. Specifically, recall that $P^*$ is lower bounded by the Gelbrich distance between $(m_X,\Sigma_X)$ and $(m_{X^*},\Sigma_{X^*})$, which is given by $(G^*)^2\triangleq\mathrm{Tr}\{\Sigma_{X}+\Sigma_{X^*}-2(\Sigma_{X}^{1/2}\Sigma_{X^*}\Sigma_{X}^{1/2})^{1/2}\}$ (see \eqref{eq:W2_greater_G2}). Given approximations for $D^*$ and $G^*$, we can approximate a lower bound on the DP function for any $P$,
	\begin{equation}\label{eq:Thm2::LwrBound_G_Sx_S*}
	D(P)\geq D^*+[(G^*-P)_+]^{2}.
	\end{equation}
	The bound is attained when $X$ and $Y$ are jointly Gaussian.
	
	\paragraph{Uniqueness} A remark is in place regarding the uniqueness of an estimator achieving \eqref{eq:Thm:DP_function::Main_result}. As we discuss below, what defines an optimal estimator $\hat{X}$ is its joint distribution with $X^*$. This joint distribution may not be unique, in which case the optimal estimator is not unique. Moreover, even if $p_{\hat{X} X^{*}}$ is unique, the uniqueness of the estimator is not guaranteed because there may be different conditional distributions $p_{\hat{X}|Y}$ that lead to the same optimal $p_{\hat{X} X^{*}}$. In other words, given the optimal $p_{\hat{X} X^{*}}$, one can choose any joint probability $p_{\hat{X} Y X^* }$ that has marginals $p_{\hat{X} X^{*}}$ and $p_{YX^* }$. One option is to take the estimator $\hat{X}$ to be a (possibly stochastic) transformation of $X^*$, namely $p_{\hat{X}|Y}=p_{\hat{X}|X^*}p_{X^*|Y}$. But there may be other options.
	In cases where either $Y$ or $\hat{X}$ are a deterministic transformation of $X^*$ (\emph{e.g.} when $X^*$ has a density, or is an invertible function of $Y$), there is a unique joint distribution $p_{\hat{X} Y X^{*}}$ with the given marginals \cite[Lemma 5.3.2]{ambrosio2008gradient}. In this case, if $p_{\hat{X}X^{*}}$ is unique then so is the estimator $p_{\hat{X}|Y}$.
	
	\paragraph{Randomness} 	Under the settings of image restoration, many methods encourage diversity in their output by adding randomness \citep{lugmayr2020srflow,bahat2020explorable,prakash2021removing}. In our setting, we may ask under what conditions there exists an optimal estimator $\hat X$ which is a deterministic function of $Y$. For example, when $p_Y=\delta_0$ but $X$ has some non-atomic distribution, it is clear that no deterministic function of $Y$ can attain perfect perceptual quality. It turns out that a sufficient condition for the optimal $\hat X$ to be a deterministic function of $Y$ is that $X^*$  have a density. We discuss this in App.~\ref{APPsec::proofs} and explicitly illustrate it in the Gaussian case (see Sec.~\ref{Sec::GaussianSetting}), where if $X^*$ has a non-singular covariance matrix then $\hat X$ is a deterministic function of~$Y$. 

	\paragraph{When is posterior sampling optimal?} Many recent image restoration methods attempt to produce diverse high perceptual quality reconstructions by sampling from the posterior distribution $p_{X|Y}$ \citep{friedman2021posterior, ohayon2021high,kawar2021stochastic}. As discussed in \citep{blau2018perception}, the posterior sampling estimator attains a perception index of $0$ (namely $W_{2}(p_X,p_{\hat{X}})=0$) and distortion $2D^*$. But an interesting question is: when is this strategy optimal? In other words, in what cases do we have that the DP function at $P=0$ equals precisely $2D^*$ and is not strictly smaller? Note from the definition of the Wasserstein distance \eqref{eq:Wp_def}, that $(P^*)^2=W_{2}^{2}(p_X,p_{X^{*}})\leq \EE[\|X-X^*\|^2]=D^*$. Using this in \eqref{eq:Thm:DP_function::Main_result}
 shows that the DP function at $P=0$ is upper bounded by
	\begin{equation}\label{eq:Thm2::UprBound_2Dmin}
	D(0)=D^*+(P^*)^2\leq2D^*,
	\end{equation}
	and the upper bound is attained when $(P^*)^2=D^*$. To see when this happens, observe that
	\begin{equation}\label{eq:Sandwitch}
	    \mathrm{Tr}\left\{\Sigma_{X}+\Sigma_{X^*}-2(\Sigma_{X}^{\frac{1}{2}}\Sigma_{X^*}\Sigma_{X}^{\frac{1}{2}})^{\frac{1}{2}}\right\}=(G^*)^2\leq (P^*)^2\leq D^*=\mathrm{Tr}\{\Sigma_X-\Sigma_{X^*}\}.
	\end{equation}
	We can see that when $\mathrm{Tr}\{\Sigma_{X^*}\}=\mathrm{Tr}\{(\Sigma_{X}^{1/2}\Sigma_{X^*}\Sigma_{X}^{1/2})^{1/2}\}$, the leftmost and rightmost sides become equal, and thus $(P^*)^2=D^*$. To understand the meaning of this condition, let us focus on the case where $\Sigma_{X}$ and $\Sigma_{X^*}$ are jointly diagonalizable. This is a reasonable assumption for natural images, where shift-invariance induces diagonalization by the Fourier basis \citep{unser1984approximation}. In this case, the condition can be written in terms of the eigenvalues of the matrices, namely $\sum_{i}\lambda_i(\Sigma_{X^*})=\sum_{i} \sqrt{\lambda_i(\Sigma_{X^*})\lambda_i(\Sigma_{X})}$. This condition is satisfied when each  $\lambda_i(\Sigma_{X^*})$ equals either $\lambda_i(\Sigma_{X})$ or $0$. Namely, the $i$th eigenvalue of the error covariance of $X^*$, which is given by $\Sigma_X-\Sigma_{X^*}$, is either $\lambda_i(\Sigma_{X})$ or $0$. We conclude that posterior sampling is optimal when there exists a subspace $\mathcal{S}$ spanned by some of the eigenvectors of $\Sigma_X$, such that the projection of $X$ onto $\mathcal{S}$ can be recovered from $Y$ with zero error, but the projection of $X$ onto $\mathcal{S}^\perp$ cannot be recovered at all (the optimal estimator is trivial). This is likely not the case in most practical scenarios. Therefore, it seems that \emph{posterior sampling is often not optimal}. That is, posterior sampling can be improved upon in terms of MSE without any sacrifice in perceptual quality.
	
	\subsection{Optimal estimators}
	\label{sec::optimalestimators}
	While Theorem~\ref{Thm:=00005Bthe-Distortion-Perception-funct} reveals the shape of the DP function, it does not provide a recipe for constructing optimal estimators on the DP tradeoff. We now discuss the nature of such estimators.
	
	Our first observation is that since $\hat{X}$ is independent of $X$ given $Y$, its MSE can be decomposed as $\EE[\|X-\hat{X}\|^2]=\EE[\|X-X^*\|^2+\EE[\|X^*-\hat{X}\|^2]$ (see App.~\ref{APPsec::proofs}). Therefore, the DP function \eqref{eq:MSE::D_P::Definition} can be equivalently written as
	\begin{equation}\label{eq:MSE::D_P::Alternative}
    D(P)=D^*+\min_{p_{\hat{X}|Y}}\left\{\EE[\|\hat{X}-X^*\|^2] \;:\; W_2(p_X,p_{\hat{X}})\leq P\right\}.
	\end{equation}
	Note that the objective in \eqref{eq:MSE::D_P::Alternative} depends on the MSE between $\hat{X}$ and $X^*$, so that we can perform the minimization on $p_{\hat{X}|X^*}$ rather than on $p_{\hat{X}|Y}$ (once we determine the optimal $p_{\hat{X}|X^*}$ we can construct a consistent $p_{\hat{X}|Y}$ as discussed above).
	
	Now, let us start by examining the leftmost side of the curve $D(P)$, which corresponds to a perfect perceptual quality estimator (\emph{i.e.}~$P=0$). In this case, the constraint becomes $p_{\hat{X}}=p_X$. Therefore, 
	\begin{equation}
    D(0)=D^*+\min_{p_{\hat{X}X^*}}\left\{\EE[\|\hat{X}-X^*\|^2] \;:\; p_{\hat{X}X^*}\in\Pi(p_X,p_{X^*})\right\},
	\end{equation}
	where $\Pi(p_X,p_{X^*})$ is the set of all probabilities on $\R^{n_x}\times\R^{n_x}$ with marginals $p_X,p_{X^*}$. One may readily recognize this as the optimization problem underlying the Wasserstein-2 distance between $p_X$ and $p_{X^*}$. This leads us to the following conclusion.
	\begin{thm}[Optimal estimator for $P=0$]
	\label{thm:perfect_perception}Let $\hat{X}_0$ be an estimator achieving perception index $0$ and MSE $D(0)$. Then its joint distribution with $X^*$ attains the optimum in the definition of $W_2(p_X,p_{X^*})$. Namely, $p_{\hat{X}_0 X^*}$ is an optimal plan between $p_X$ and $p_{X^*}$.
	\end{thm}
	
	Having understood the estimator $\hat{X}_0$ at the leftmost end of the tradeoff, we now turn to study optimal estimators for arbitrary $P$. Interestingly, we can show that Problem \eqref{eq:MSE::D_P::Alternative} is equivalent to (see App.~\ref{APPsec::proofs})
	\begin{equation}\label{eq:ObjectiveArbitraryP}
	D(P)=D^*+\min_{p_{\hat{X}}}\left\{ W^2_2(p_{\hat{X}},p_{X^{*}})\;:\; W_{2}(p_{\hat{X}},p_X)\leq P\right\}.
	\end{equation}
	Namely, an optimal $p_{\hat{X}}$ is closest to $p_{X^*}$ among all distributions within a ball of radius~$P$ around $p_X$, as illustrated in Fig.~\ref{fig:geodesic}. 
	Moreover, $p_{\hat{X}X^*}$ is an optimal plan between $p_{\hat{X}}$ and $p_{X^*}$.
	As it turns out, this somewhat abstract viewpoint leads to a rather practical construction for $\hat{X}$ from the estimators $\hat{X}_0$ and $X^*$ at the two extremes of the tradeoff.  Specifically, we have the following result, proved in App.~\ref{APPsec::proofs}.
	\begin{thm}[Optimal estimators for arbitrary $P$]
		\label{Thm::extrapol}Let $\hat{X}_0$ be an estimator achieving perception index~$0$ and MSE $D(0)$. Then for any $P\in[0,P^*]$,
		the estimator
		\begin{equation}
		\hat{X}_P=\left(1-\frac{P}{P^*}\right)\hat{X}_{0}+\frac{P}{P^*}X^{*}\label{eq:hat_X::extrapolated}
		\end{equation}
		is optimal for perception index $P$. Namely, it achieves perception index $P$ and distortion $D(P)$.
	\end{thm}
	
	Theorem~\ref{Thm::extrapol} has important implications for perceptual signal restoration. For example, in the task of image super-resolution, there exist many deep network based methods that achieve a low MSE \citep{lim2017enhanced,ulyanov2018deep,shocher2018zero}. These provide an approximation for $X^*$. Moreover, there is an abundance of methods that achieve good perceptual quality at the price of a reasonable degradation in MSE (often by incorporating a GAN-based loss) \citep{ledig2017photo,wang2018esrgan,shaham2019singan}. These constitute approximations for $\hat{X}_0$. However, achieving results that strike other prescribed balances between MSE and perceptual quality commonly require training a different model for each setting. \citet{shoshan2019dynamic} and \citet{navarrete2018multi} tried to address this difficulty by introducing new training techniques that allow traversing the distortion-perception tradeoff at test time. But, interestingly, Theorem~\ref{Thm::extrapol} shows that in our setting such specialized training methods are not required. Having a model that leads to low MSE and one that leads to good perceptual quality, it is possible to construct any other estimator on the DP tradeoff, by simply averaging the outputs of these two models with appropriate weights. We illustrate this in Sec.~\ref{sec::numerical}. 

	\subsection{The Gaussian setting}
	\label{Sec::GaussianSetting}
	When $X$ and $Y$ are jointly Gaussian, it is well known that the minimum MSE estimator $X^*$ is a linear function of the measurements $Y$. However, it is not \emph{a-priori} clear whether all estimators along the DP tradeoff are linear in this case, and what kind of randomness they possess. As we now show, equipped with Theorem~\ref{Thm::extrapol}, we can obtain closed form expressions for optimal estimators for any $P$. For simplicity, we assume here that $X$ and $Y$ have zero means and that $\Sigma_X,\Sigma_{Y} \succ0$.

	It is instructive to start by considering the simple case, where $\Sigma_{X^*}$ is non-singular (in Theorem~\ref{thm:Gaussian1} below we address the more general case of a possibly singular $\Sigma_{X^*}$). It is well known that 
	\begin{equation}\label{eq:SigmaXstar}
	    X^*=\Sigma_{XY}\Sigma_Y^{-1}Y,\qquad \Sigma_{X^*}=\Sigma_{XY}\Sigma_{Y}^{-1}\Sigma_{YX}.
	\end{equation}
	
	Now, since we assumed that $\Sigma_X,\Sigma_{X^*}\succ0$, we have from Theorem~\ref{thm:perfect_perception} and \eqref{eq:optimal_dist_for_gauss},\eqref{eq::T1->2:def} that
	\begin{equation}
	    \hat{X}_0 = \Sigma_{X^*}^{-\frac{1}{2}}\left(\Sigma_{X^*}^{\frac{1}{2}}\Sigma_{X}\Sigma_{X^*}^{\frac{1}{2}}\right)^{\frac{1}{2}}\Sigma_{X^*}^{-\frac{1}{2}} X^*.
	\end{equation} 
	Finally, we know that $P^*=G^*$, which is given by the left-hand side of \eqref{eq:Sandwitch}. Substituting these expressions into \eqref{eq:hat_X::extrapolated}, we obtain that an optimal estimator for perception $P\in[0,G^*]$ is given by
	\begin{equation}\label{eq:GaussXpInvertible}
	   \hat{X}_P=\left(\left(1-\frac{P}{G^*}\right)\Sigma_{X^*}^{-\frac{1}{2}}\left(\Sigma_{X^*}^{\frac{1}{2}}\Sigma_{X}\Sigma_{X^*}^{\frac{1}{2}}\right)^{\frac{1}{2}}\Sigma_{X^*}^{-\frac{1}{2}} +\frac{P}{G^*}I\right) \Sigma_{XY}\Sigma_Y^{-1}Y.
	\end{equation}
	As can be seen, this optimal estimator is a deterministic linear transformation of $Y$ for any $P$.
	
	The setting just described does not cover the case where $Y$ is of lower dimensionality than $X$ because in that case $\Sigma_{X^*}$ is necessarily singular (it is a $n_x\times n_x$ matrix of rank at most $n_y$; see \eqref{eq:SigmaXstar}). In this case, any deterministic linear function of $Y$ would result in an estimator $\hat{X}$ with a rank-$n_y$ covariance. Obviously, the distribution of such an estimator cannot be arbitrarily close to that of $X$, whose covariance has rank $n_x$. What is the optimal estimator in this more general setting, then? 
	
	\begin{thm}[Optimal estimators in the Gaussian case]\label{thm:Gaussian1}
		Assume $X$ and $Y$ are zero-mean jointly Gaussian random vectors with $\Sigma_X,\Sigma_{Y} \succ0$. Denote $T^{*}\triangleq T_{p_X\rightarrow p_{X^*}}=\Sigma_{X}^{-1/2}(\Sigma_{X}^{1/2}\Sigma_{X^*}\Sigma_{X}^{1/2})^{1/2}\Sigma_{X}^{-1/2}$. 
		Then for any $P\in[0,G^{*}]$, an estimator with perception index $P$ and MSE $D(P)$ can be constructed as
		\begin{equation} \label{eq::X_P_Gauss:interpolation}
		\hat{X}_P=\left(\left(1-\frac{P}{G^{*}}\right)\Sigma_{X}^{\frac{1}{2}}\left(\Sigma_{X}^{\frac{1}{2}}\Sigma_{X^*}\Sigma_{X}^{\frac{1}{2}}\right)^{\frac{1}{2}}\Sigma_{X}^{-\frac{1}{2}}\Sigma_{X^*}^{\dagger}+\frac{P}{G^{*}}I\right)\Sigma_{XY}\Sigma_{Y}^{-1}Y+\left(1-\frac{P}{G^{*}}\right)W,
		\end{equation}
		where $W$ is a zero-mean Gaussian noise with covariance $\Sigma_{W}=\Sigma_{X}^{1/2}(I-\Sigma_{X}^{1/2}T^{*}\Sigma_{X^*}^{\dagger}T^{*}\Sigma_{X}^{1/2})\Sigma_{X}^{1/2}$, which is independent of $Y,X$, and $\Sigma_{X^*}^{\dagger}$ is the pseudo-inverse of $\Sigma_{X^*}$.
	\end{thm}
	Note that in this case, we indeed have a random noise component that shapes the covariance of $\hat{X}_P$ to become closer to $\Sigma_X$ as $P$ gets closer to $0$. It can be shown (see App.~\ref{APPsec::proofs}) that when $\Sigma_{X^*}$ is invertible, $\Sigma_W=0$ and \eqref{eq::X_P_Gauss:interpolation} reduces to \eqref{eq:GaussXpInvertible}. Also note that, as in \eqref{eq:GaussXpInvertible}, the dependence of $\hat{X}_P$ on $Y$ in \eqref{eq::X_P_Gauss:interpolation} is only through $X^*=\Sigma_{XY}\Sigma_{Y}^{-1}Y$.
	
	As mentioned in Sec.~\ref{sec::The MSE--Wasserstein-2 tradeoff}, the optimal estimator is generally not unique. Interestingly, in the Gaussian setting we can explicitly characterize a \emph{set} of optimal estimators. 
	\begin{thm}[A set of optimal estimators in the Gaussian case] \label{thm:Gaussian_not_unique}
		Consider the setting of Theorem~\ref{thm:Gaussian1}. Let $\Sigma_{\hat X_0 Y}\in \R^{n_{x}\times n_{y}}$ satisfy
		\begin{equation}
		\Sigma_{\hat X_0 Y}\Sigma_{Y}^{-1}\Sigma_{YX}=\Sigma_{X}^{\frac{1}{2}}(\Sigma_{X}^{\frac{1}{2}}\Sigma_{X^*}\Sigma_{X}^{\frac{1}{2}})^{\frac{1}{2}}\Sigma_{X}^{-\frac{1}{2}},\label{eq:Thm_Gaussian_General::M_cond1-2-1-1}
		\end{equation}
		and $W_{0}$ be a zero-mean Gaussian noise with covariance
		\begin{equation}
		\Sigma_{W_{0}}=\Sigma_{X}-\Sigma_{\hat X_0 Y}\Sigma_{Y}^{-1}\Sigma_{\hat X_0 Y}^{T}\succeq0
		\label{eq:Thm_Gaussian_General::M_cond2-2-1-1}
		\end{equation}
		that is independent of $X,Y$. Then, for any $P\in[0,G^*]$, an optimal estimator with perception index $P$ can be obtained by
		\begin{equation}
		\hat{X}_P=\left(\left(1-\frac{P}{G^{*}}\right)\Sigma_{\hat X_0 Y}+\frac{P}{G^{*}}\Sigma_{XY}\right)\Sigma_{Y}^{-1}Y+\left(1-\frac{P}{G^{*}}\right)W_{0}.
		\end{equation} The estimator given in \eqref{eq::X_P_Gauss:interpolation}
		is one solution to \eqref{eq:Thm_Gaussian_General::M_cond1-2-1-1}-\eqref{eq:Thm_Gaussian_General::M_cond2-2-1-1},
		but is generally not unique.
	\end{thm}
	

\section{A geometric perspective on the distortion-perception tradeoff}
\label{sec::geometric}
	In this section we provide a geometric point of view on our main results. Specifically, we show that the results of Theorems \ref{Thm:=00005Bthe-Distortion-Perception-funct} and \ref{Thm::extrapol} are a consequence
	of a more general geometric property of the space $\mathcal{W}_{2}(\R^{n_{x}})$.
	In the Gaussian case, this is simplified to a geometry of covariance
	matrices.
	
	Recall from \eqref{eq:ObjectiveArbitraryP} that the optimal $p_{\hat{X}}$ is the one closest to $p_{X^*}$ (in terms of Wasserstein distance) among all measures at a distance $P$ from $p_X$. This implies that to determine $p_{\hat{X}}$, we should traverse the geodesic between $p_{X^*}$ and $p_{X}$ until reaching a distance of $P$ from $p_X$.	Furthermore, $p_{\hat{X}X^*}$ should be the optimal plan between $p_{\hat{X}}$ and $p_{X^*}$. Interestingly, geodesics in Wasserstein spaces take a particularly simple form, and their explicit construction also turns out to satisfy the latter requirement. Specifically,
	let $\gamma,\mu$ be measures in $\mathcal{W}_{2}(\R^{d})$, let $\nu\in\Pi(\gamma,\mu)$	be an optimal plan attaining $W_2(\gamma,\mu)$, and let $\pi_{i}$ denote the projection $\pi_{i}:\R^{d}\times \R^{d}\rightarrow \R^{d}$
	such that $\pi_{i}((x_{1},x_{2}))=x_{i},\;i=1,2$. Then, the curve
	\begin{equation}
	\gamma_{t}\triangleq\left[(1-t)\pi_{1}+t\pi_{2}\right]\#\nu,\quad t\in[0,1]\label{eq:Geodesic:nu}
	\end{equation}
	is a constant-speed geodesic from $\gamma$ to $\mu$ in $\mathcal{W}_{2}(\R^{d})$
	\citep{ambrosio2008gradient}, where $\#$ is the push-forward operation\footnote{For measures $\gamma,\mu$ on $\mathcal{X},\mathcal{Y}$, we say that a measurable transform $T:\mathcal{X}\rightarrow\mathcal{Y}$ pushes $\gamma$ forward to $\mu$ (denoted  $T\#\gamma=\mu$) iff $\gamma(T^{-1}(B))=\mu(B)$ for any measurable $B \subseteq \mathcal{Y}.$}. Particularly, 
	\begin{equation}
	W_{2}(\gamma_{t},\gamma_{s})=|t-s|W_{2}(\gamma,\mu),
	\end{equation}
	and it follows that $\ensuremath{W_{2}(\gamma_{t},\gamma)=tW_{2}(\gamma,\mu)}$
	and $\ensuremath{W_{2}(\gamma_{t},\mu)=(1-t)W_{2}(\gamma,\mu)}$. Furthermore, if $\gamma_t,t\in [0,1]$ is a constant-speed geodesic with $\gamma_0=\gamma,\gamma_1=\mu$, then the optimal plans between $\gamma,\gamma_t$ and between $\gamma_t,\mu$ are given by \begin{equation}
	    \left[\pi_1,(1-t)\pi_{1}+t\pi_{2}\right]\#\nu,\quad \left[(1-t)\pi_{1}+t\pi_{2},\pi_2\right]\#\nu,
	\end{equation}
	respectively, where $\nu\in\Pi(\gamma,\mu)$ is some optimal plan.
	Applying \eqref{eq:Geodesic:nu} to $(\hat X_0,X^*)\sim \nu$ with $t=P/P^*$, we obtain \eqref{eq:hat_X::extrapolated}, where we show that the obtained estimator achieves $\EE[\|\hat {X}_P-X^*\|^2]=(1-t)^2 W^2_2(p_X,p_{X^*})$. This explains the result of Theorem \ref{Thm::extrapol}. 

	It is worth mentioning that this geometric interpretation is simplified under some common settings. For example, when $\gamma$ is absolutely continuous (w.r.t.~the Lebesgue measure),
	we have a measurable map $T_{\gamma\rightarrow\mu}$ which is
	the solution to the optimal transport problem with the quadratic cost \citep[Thm 1.6.2, p.16]{panaretos2020invitation}. The geodesic \eqref{eq:Geodesic:nu} then takes the form
	\begin{equation}
	\gamma_{t}=[Id+t(T_{\gamma\rightarrow\mu}-Id)]\#\gamma,\quad t\in[0,1].
	\end{equation}
	Therefore, in our setting, if $\gamma=p_{X^*}$ has a density, then we can obtain $\hat {X}_P$ by the deterministic transformation $[X^*+\left( 1-\frac{P}{P^*} \right) \left(T_{p_{X^*}\rightarrow p_{X}}(X^*)-X^*\right)]$ (see Remark about randomness in Sec.~\ref{sec::The MSE--Wasserstein-2 tradeoff}).
	
	Further simplification arises when $\gamma,\mu$ are centered non-singular Gaussian measures, in which case $T_{\gamma\rightarrow\mu}$
	is the linear and symmetric transformation \eqref{eq::T1->2:def}.
	Then, $\gamma_{t}$ is a Gaussian measure with covariance $\Sigma_{\gamma_{t}}=T_{t}\Sigma_{\gamma}T_{t}$,
	where $T_{t}\triangleq[I+t(T_{\gamma\rightarrow\mu}-I)].$
	Therefore, in the Gaussian case, the shortest path (\ref{eq:Geodesic:nu})
	between distributions is reduced to a trajectory in the geometry of covariance matrices induced by the Gelbrich distance \citep{takatsu2010wasserstein}. If additionally $\Sigma_{\gamma}$ and $\Sigma_{\mu}$ commute, then the Gelbrich distance is further reduced to the $\ell^2$-distance between matrices, as we discuss in App.~\ref{appsec::commute}.

\section{Numerical illustration} \label{sec::numerical}

In this Section we evaluate $12$ super resolution algorithms on the BSD100 dataset\footnote{All codes are freely available and provided by the authors. The BSD100 dataset is free to download for non-commercial research.}
 \citep{martin2001database}.  The evaluated algorithms include EDSR \citep{lim2017enhanced}, ESRGAN \citep{wang2018esrgan}, SinGAN \citep{shaham2019singan}, ZSSR \citep{shocher2018zero}, DIP \citep{ulyanov2018deep}, SRResNet variants which optimize MSE and VGG$_{2,2}$, SRGAN variants which optimize MSE, VGG$_{2,2}$ and VGG$_{5,4}$ in addition to an adversarial loss \citep{ledig2017photo}, ENet \citep{sajjadi2017enhancenet} (``PAT'' and ``E'' variants). Low resolution images were obtained by $ 4\times$ downsampling using a bicubic kernel.

In Figure \ref{fig:DG} we plot each method on the distortion-perception plane. Specifically, we consider natural (and reconstructed) images to be stationary random sources, and use $9 \times 9 $ patches  (totally $1.6\times 10^6$ patches) from the RGB images to empirically estimate the mean and covariance matrix for the ground-truth images, and for the reconstructions produced by each method. We then use the estimated Gelbrich distances \eqref{eq:Gelbrich_dist} between the patch distribution of each method and that of ground-truth images, as a perceptual quality index. Recall this is a lower bound on the Wasserstein distance.

We consider EDSR \cite{lim2017enhanced} to be the best MSE estimator $X^*$ since it achieves the lowest distortion among the evaluated methods. We therefore estimate the lower bound \eqref{eq:Thm2::LwrBound_G_Sx_S*} as 
\[\hat D (P)=D_{\text{EDSR}}+\left[(P_{\text{EDSR}}-P)_+\right]^2,\]
where $D_{\text{EDSR}}$ is the MSE of EDSR, and $P_{\text{EDSR}}$ is the estimated Gelbrich distance between EDSR reconstructions and ground-truth images. Note the unoccupied region under the estimated curve in Figure~\ref{fig:DG}, which is indeed unattainable according to the theory.
	
We also present 11 estimators $\hat X_t$ which we construct by interpolation between EDSR and ESRGAN \cite{wang2018esrgan},
	$ \hat X_t = tX_{\text{EDSR}} + (1-t)X_{\text{ESRGAN}}$.
	 We observe (Figure \ref{fig:DG}) that estimators constructed using these two extreme points are closer to the optimal DP tradeoff than the evaluated methods. Also note that since ESRGAN does not attain $0$-perception index, we are practically able to use negative values $t<0$ to extrapolate better perception-quality estimators $\hat X_{-0.05}$ and $\hat X_{-0.1}$. In Figure \ref{fig:SrganVSintp} we present a visual comparison between SRGAN-VGG$_{2,2}$ \citep{ledig2017photo}
	and our interpolated estimator $\hat X_{0.12}$. Both achieve roughly the same RMSE distortion ($18.09$ for SRGAN, $18.15$ for $\hat X_{0.12}$), but our estimator achieves a lower perception index. Namely, by using interpolation, we manage to achieve improvement in perceptual quality, without degradation in distortion. The improvement in visual quality is also apparent in the figure.	Additional visual comparisons can be found in the Appendix.

	\begin{figure}[h]
	\begin{centering}
			\includegraphics[bb=10bp 0bp 460bp 307bp,clip,scale=0.85,width=0.80\linewidth]{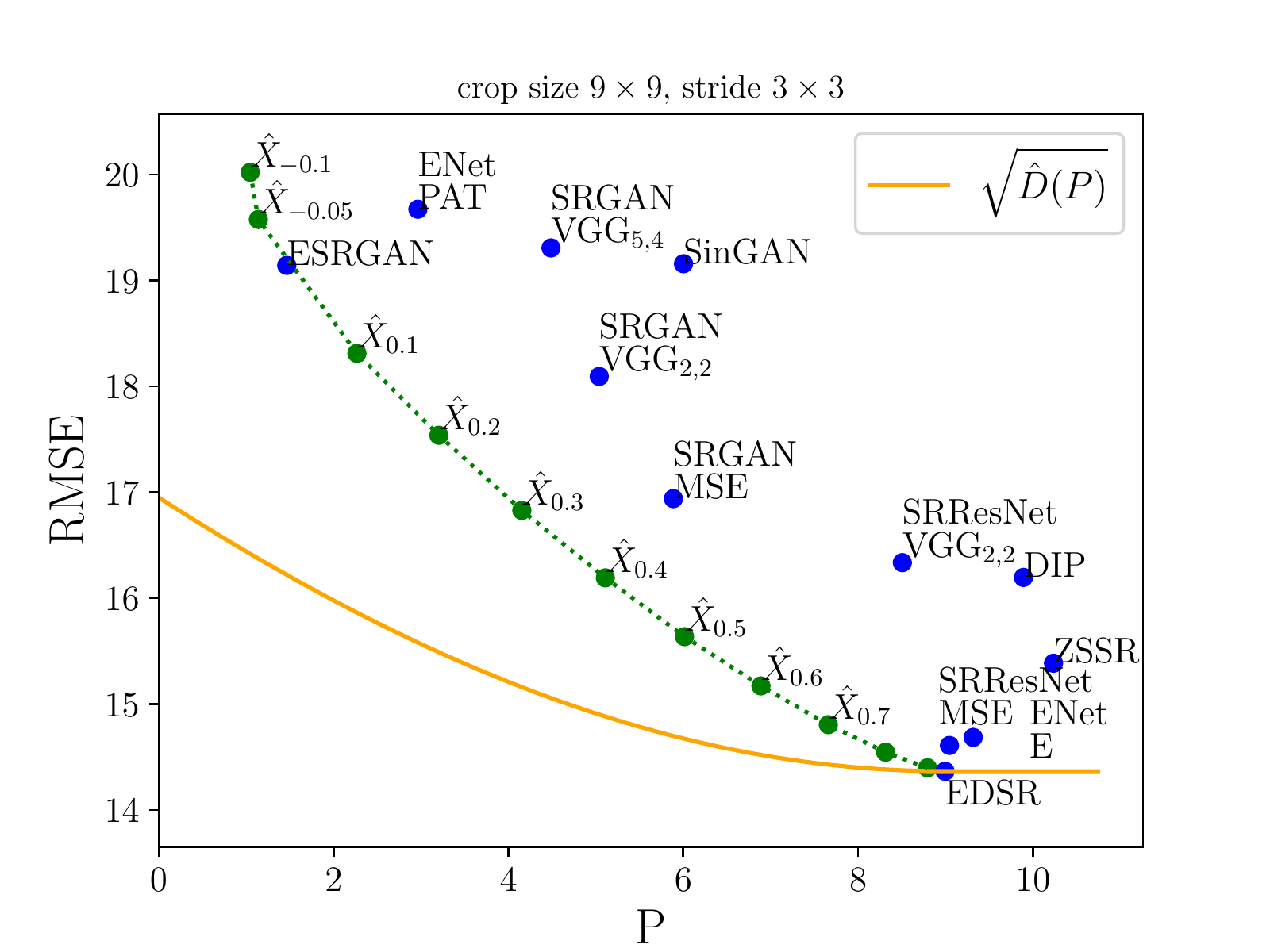}
			\par\end{centering}
	
		\centering{}\caption{\label{fig:DG} \textbf{Evaluation of SR algorithms}. We plot 12 algorithms (Blue) on the Distortion-Perception plane. Here we estimate perception using the Gelbrich distance between empirical means and covariances of the original data and reconstructed data. $\hat D(P)$ (Orange) is the estimated lower bound \eqref{eq:Thm2::LwrBound_G_Sx_S*} where we consider EDSR to be the global minimizer $X^*$. Note the unoccupied region under the estimated curve, which is unattainable.
			We also plot 11 estimators $\hat X_t$ (Green) created by an interpolation between EDSR and ESRGAN, using different relative weights $t$. Note that estimators constructed using these two extreme estimators are closer to the optimal DP curve than the compared methods.}
	\end{figure}

	\begin{figure}[]
	\begin{centering}
			\includegraphics[bb=12bp 200bp 1150bp 522bp,clip,scale=0.99,width=.95\linewidth]{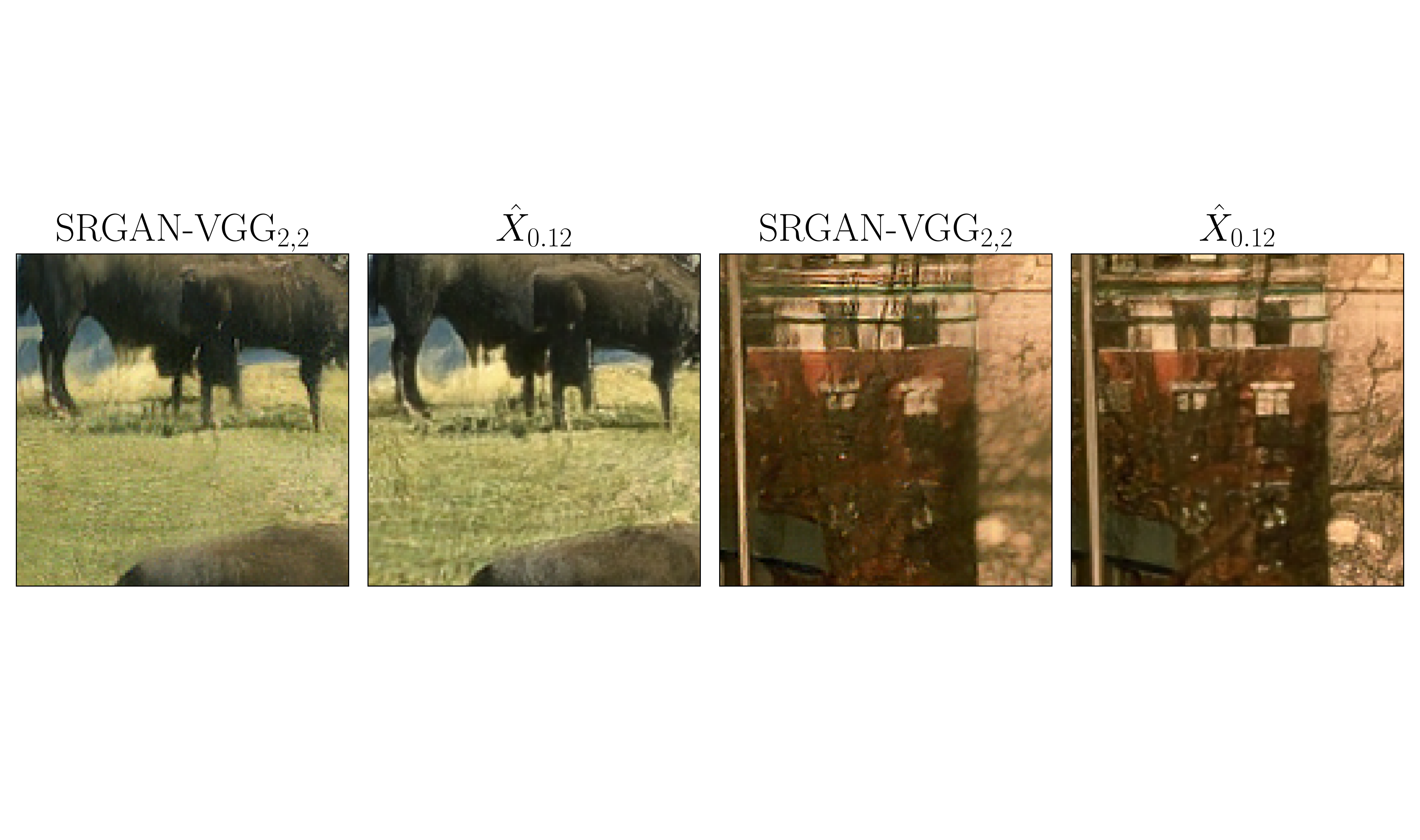}
		\par\end{centering}
	
	\centering{}\caption{ \label{fig:SrganVSintp} 
		\textbf{A visual comparison between estimators with approximately the same MSE}. Left: SRGAN-VGG$_{2,2}$. Right: $\hat X_{0.12}$, an interpolation between EDSR and ESRGAN using $t=0.12$. Observe the improvement in perceptual quality, without any significant degradation in distortion.}
	\end{figure}

\section{Conclusion}
In this paper we provide a full characterization of the distortion-perception tradeoff for the MSE distortion and the Wasserstein-$2$ perception index. We show that optimal estimators are obtained by interpolation between the minimum MSE estimator and an optimal perfect perception quality estimator. In the Gaussian case, we explicitly formulate these estimators. To the best of our knowledge, this is the first work to derive such closed-form expressions. Our work paves the way towards fully understanding the DP tradeoff under more general distortions and perceptual criteria, and bridging between fidelity and visual quality at test-time, without training different models.

\setcitestyle{numbers}
\bibliography{references}

\begin{thebibliography}{32}
\providecommand{\natexlab}[1]{#1}
\providecommand{\url}[1]{\texttt{#1}}
\expandafter\ifx\csname urlstyle\endcsname\relax
  \providecommand{\doi}[1]{doi: #1}\else
  \providecommand{\doi}{doi: \begingroup \urlstyle{rm}\Url}\fi

\bibitem[Abid et~al.(2021)Abid, Hedhli, and Gagn{\'e}]{abid2021generative}
Mohamed~Abderrahmen Abid, Ihsen Hedhli, and Christian Gagn{\'e}.
\newblock A generative model for hallucinating diverse versions of super
  resolution images.
\newblock \emph{arXiv preprint arXiv:2102.06624}, 2021.

\bibitem[Ambrosio et~al.(2008)Ambrosio, Gigli, and
  Savar{\'e}]{ambrosio2008gradient}
Luigi Ambrosio, Nicola Gigli, and Giuseppe Savar{\'e}.
\newblock \emph{Gradient flows: in metric spaces and in the space of
  probability measures}.
\newblock Springer Science \& Business Media, 2008.

\bibitem[Bahat and Michaeli(2020)]{bahat2020explorable}
Yuval Bahat and Tomer Michaeli.
\newblock Explorable super resolution.
\newblock In \emph{Proceedings of the IEEE/CVF Conference on Computer Vision
  and Pattern Recognition}, pages 2716--2725, 2020.

\bibitem[Blau and Michaeli(2018)]{blau2018perception}
Yochai Blau and Tomer Michaeli.
\newblock The perception-distortion tradeoff.
\newblock In \emph{Proceedings of the IEEE Conference on Computer Vision and
  Pattern Recognition}, pages 6228--6237, 2018.

\bibitem[Blau and Michaeli(2019)]{blau2019rethinking}
Yochai Blau and Tomer Michaeli.
\newblock Rethinking lossy compression: The rate-distortion-perception
  tradeoff.
\newblock In \emph{International Conference on Machine Learning}, pages
  675--685. PMLR, 2019.

\bibitem[Deng(2018)]{deng2018enhancing}
Xin Deng.
\newblock Enhancing image quality via style transfer for single image
  super-resolution.
\newblock \emph{IEEE Signal Processing Letters}, 25\penalty0 (4):\penalty0
  571--575, 2018.

\bibitem[Friedman and Weiss(2021)]{friedman2021posterior}
Roy Friedman and Yair Weiss.
\newblock Posterior sampling for image restoration using explicit patch priors.
\newblock \emph{arXiv preprint arXiv:2104.09895}, 2021.

\bibitem[Gelbrich(1990)]{gelbrich1990formula}
Matthias Gelbrich.
\newblock On a formula for the l2 wasserstein metric between measures on
  euclidean and hilbert spaces.
\newblock \emph{Mathematische Nachrichten}, 147\penalty0 (1):\penalty0
  185--203, 1990.

\bibitem[Gray(2006)]{gray2006toeplitz}
Robert~M Gray.
\newblock Toeplitz and circulant matrices: A review.
\newblock 2006.

\bibitem[Kawar et~al.(2021)Kawar, Vaksman, and Elad]{kawar2021stochastic}
Bahjat Kawar, Gregory Vaksman, and Michael Elad.
\newblock Stochastic image denoising by sampling from the posterior
  distribution.
\newblock \emph{arXiv preprint arXiv:2101.09552}, 2021.

\bibitem[Knott and Smith(1984)]{knott1984optimal}
Martin Knott and Cyril~S Smith.
\newblock On the optimal mapping of distributions.
\newblock \emph{Journal of Optimization Theory and Applications}, 43\penalty0
  (1):\penalty0 39--49, 1984.

\bibitem[Ledig et~al.(2017)Ledig, Theis, Husz{\'a}r, Caballero, Cunningham,
  Acosta, Aitken, Tejani, Totz, Wang, et~al.]{ledig2017photo}
Christian Ledig, Lucas Theis, Ferenc Husz{\'a}r, Jose Caballero, Andrew
  Cunningham, Alejandro Acosta, Andrew Aitken, Alykhan Tejani, Johannes Totz,
  Zehan Wang, et~al.
\newblock Photo-realistic single image super-resolution using a generative
  adversarial network.
\newblock In \emph{Proceedings of the IEEE conference on computer vision and
  pattern recognition}, pages 4681--4690, 2017.

\bibitem[Lim et~al.(2017)Lim, Son, Kim, Nah, and Mu~Lee]{lim2017enhanced}
Bee Lim, Sanghyun Son, Heewon Kim, Seungjun Nah, and Kyoung Mu~Lee.
\newblock Enhanced deep residual networks for single image super-resolution.
\newblock In \emph{Proceedings of the IEEE conference on computer vision and
  pattern recognition workshops}, pages 136--144, 2017.

\bibitem[Liu et~al.(2019)Liu, Zhang, and Xiong]{liu2019classification}
Dong Liu, Haochen Zhang, and Zhiwei Xiong.
\newblock On the classification-distortion-perception tradeoff.
\newblock \emph{arXiv preprint arXiv:1904.08816}, 2019.

\bibitem[Lugmayr et~al.(2020)Lugmayr, Danelljan, Van~Gool, and
  Timofte]{lugmayr2020srflow}
Andreas Lugmayr, Martin Danelljan, Luc Van~Gool, and Radu Timofte.
\newblock Srflow: Learning the super-resolution space with normalizing flow.
\newblock In \emph{European Conference on Computer Vision}, pages 715--732.
  Springer, 2020.

\bibitem[Martin et~al.(2001)Martin, Fowlkes, Tal, and
  Malik]{martin2001database}
David Martin, Charless Fowlkes, Doron Tal, and Jitendra Malik.
\newblock A database of human segmented natural images and its application to
  evaluating segmentation algorithms and measuring ecological statistics.
\newblock In \emph{Proceedings Eighth IEEE International Conference on Computer
  Vision. ICCV 2001}, volume~2, pages 416--423. IEEE, 2001.

\bibitem[Navarrete~Michelini et~al.(2018)Navarrete~Michelini, Zhu, and
  Liu]{navarrete2018multi}
Pablo Navarrete~Michelini, Dan Zhu, and Hanwen Liu.
\newblock Multi--scale recursive and perception--distortion controllable image
  super--resolution.
\newblock In \emph{Proceedings of the European Conference on Computer Vision
  (ECCV) Workshops}, pages 0--0, 2018.

\bibitem[Ohayon et~al.(2021)Ohayon, Adrai, Vaksman, Elad, and
  Milanfar]{ohayon2021high}
Guy Ohayon, Theo Adrai, Gregory Vaksman, Michael Elad, and Peyman Milanfar.
\newblock High perceptual quality image denoising with a posterior sampling
  cgan.
\newblock \emph{arXiv preprint arXiv:2103.04192}, 2021.

\bibitem[Panaretos and Zemel(2020)]{panaretos2020invitation}
Victor~M Panaretos and Yoav Zemel.
\newblock \emph{An invitation to statistics in Wasserstein space}.
\newblock Springer Nature, 2020.

\bibitem[Prakash et~al.(2020)Prakash, Krull, and Jug]{prakash2020divnoising}
Mangal Prakash, Alexander Krull, and Florian Jug.
\newblock Divnoising: diversity denoising with fully convolutional variational
  autoencoders.
\newblock \emph{arXiv preprint arXiv:2006.06072}, 2020.

\bibitem[Prakash et~al.(2021)Prakash, Delbracio, Milanfar, and
  Jug]{prakash2021removing}
Mangal Prakash, Mauricio Delbracio, Peyman Milanfar, and Florian Jug.
\newblock Removing pixel noises and spatial artifacts with generative diversity
  denoising methods.
\newblock \emph{arXiv preprint arXiv:2104.01374}, 2021.

\bibitem[Sajjadi et~al.(2017)Sajjadi, Scholkopf, and
  Hirsch]{sajjadi2017enhancenet}
Mehdi~SM Sajjadi, Bernhard Scholkopf, and Michael Hirsch.
\newblock Enhancenet: Single image super-resolution through automated texture
  synthesis.
\newblock In \emph{Proceedings of the IEEE International Conference on Computer
  Vision}, pages 4491--4500, 2017.

\bibitem[Shaham et~al.(2019)Shaham, Dekel, and Michaeli]{shaham2019singan}
Tamar~Rott Shaham, Tali Dekel, and Tomer Michaeli.
\newblock Singan: Learning a generative model from a single natural image.
\newblock In \emph{Proceedings of the IEEE/CVF International Conference on
  Computer Vision}, pages 4570--4580, 2019.

\bibitem[Shocher et~al.(2018)Shocher, Cohen, and Irani]{shocher2018zero}
Assaf Shocher, Nadav Cohen, and Michal Irani.
\newblock “zero-shot” super-resolution using deep internal learning.
\newblock In \emph{Proceedings of the IEEE Conference on Computer Vision and
  Pattern Recognition}, pages 3118--3126, 2018.

\bibitem[Shoshan et~al.(2019)Shoshan, Mechrez, and
  Zelnik-Manor]{shoshan2019dynamic}
Alon Shoshan, Roey Mechrez, and Lihi Zelnik-Manor.
\newblock Dynamic-net: Tuning the objective without re-training for synthesis
  tasks.
\newblock In \emph{Proceedings of the IEEE/CVF International Conference on
  Computer Vision}, pages 3215--3223, 2019.

\bibitem[Takatsu(2010)]{takatsu2010wasserstein}
Asuka Takatsu.
\newblock On wasserstein geometry of gaussian measures.
\newblock In \emph{Probabilistic approach to geometry}, pages 463--472.
  Mathematical Society of Japan, 2010.

\bibitem[Ulyanov et~al.(2018)Ulyanov, Vedaldi, and Lempitsky]{ulyanov2018deep}
Dmitry Ulyanov, Andrea Vedaldi, and Victor Lempitsky.
\newblock Deep image prior.
\newblock In \emph{Proceedings of the IEEE conference on computer vision and
  pattern recognition}, pages 9446--9454, 2018.

\bibitem[Unser(1984)]{unser1984approximation}
Michael Unser.
\newblock On the approximation of the discrete karhunen-loeve transform for
  stationary processes.
\newblock \emph{Signal Processing}, 7\penalty0 (3):\penalty0 231--249, 1984.

\bibitem[Wang et~al.(2018)Wang, Yu, Wu, Gu, Liu, Dong, Qiao, and
  Change~Loy]{wang2018esrgan}
Xintao Wang, Ke~Yu, Shixiang Wu, Jinjin Gu, Yihao Liu, Chao Dong, Yu~Qiao, and
  Chen Change~Loy.
\newblock Esrgan: Enhanced super-resolution generative adversarial networks.
\newblock In \emph{Proceedings of the European Conference on Computer Vision
  (ECCV) Workshops}, pages 0--0, 2018.

\bibitem[Wang et~al.(2019)Wang, Yu, Dong, Tang, and Loy]{wang2019deep}
Xintao Wang, Ke~Yu, Chao Dong, Xiaoou Tang, and Chen~Change Loy.
\newblock Deep network interpolation for continuous imagery effect transition.
\newblock In \emph{Proceedings of the IEEE/CVF Conference on Computer Vision
  and Pattern Recognition}, pages 1692--1701, 2019.

\bibitem[Wang et~al.(2004)Wang, Bovik, Sheikh, and Simoncelli]{wang2004image}
Zhou Wang, Alan~C Bovik, Hamid~R Sheikh, and Eero~P Simoncelli.
\newblock Image quality assessment: from error visibility to structural
  similarity.
\newblock \emph{IEEE transactions on image processing}, 13\penalty0
  (4):\penalty0 600--612, 2004.

\bibitem[Zhang et~al.(2019)Zhang, Wang, and Nehorai]{zhang2019optimal}
Zhen Zhang, Mianzhi Wang, and Arye Nehorai.
\newblock Optimal transport in reproducing kernel hilbert spaces: Theory and
  applications.
\newblock \emph{IEEE transactions on pattern analysis and machine
  intelligence}, 42\penalty0 (7):\penalty0 1741--1754, 2019.

\end{thebibliography}
\bibliographystyle{plainnat}


\newpage
	\appendix
	\begin{center}
	 { \Large\bf A Theory of the Distortion-Perception Tradeoff in  
	 
	 Wasserstein Space - Supplementary Material}
	 \end{center}

	In Appendix A we present the distortion-perception tradeoff in general metric spaces. We formulate the problem of finding a perfect perception quality estimator as an optimal transportation problem, and extend some of the background provided in Sec.~\ref{Sec::D-P::optimalTransport}. In Appendix B we provide detailed proofs to the results appearing in the paper. Appendix C examines settings where covariance matrices commute. In Appendix D we discuss the details of the numerical illustrations of Sec.~\ref{sec::numerical} and provide additional visual results.

	\section{Background and extensions}\label{Appsec::Supp}
\subsection{The distortion-perception function} 

{
	In Sec.~\ref{Sec::D-P::optimalTransport} of the main text we presented the setting of Euclidean space for simplicity. For the sake of completeness, we present here a more general setup.
}

	Let $X,Y$ be random variables on separable metric spaces $\mathcal{X\mathrm{,}Y}$,
	with joint probability $p_{X,Y}$ on $\mathcal{X}\times\mathcal{Y}$. Given a distortion
	function $d:\mathcal{X}\times\mathcal{X}\rightarrow \R^{+}\cup\{0\}$, we aim to find an estimator $\hat{X}\in\mathcal{X}$ defined by
	a conditional distribution $p_{\hat{X}|Y}$ (which induces a marginal distribution $p_{\hat{X}}$), minimizing the expectation $\EE [d(X,\hat{X})]$ under the constraint $d_{p}(p_{\hat{X}},p_X)\leq P$. Here, $d_{p}$ is some divergence between probability measures.
	We further assume the Markov relation $X\rightarrow Y\rightarrow\hat{X}$,
	i.e.~$X,\hat{X}$ are independent given $Y$. 
	Similarly
	to \citet{blau2018perception} we define the distortion-perception
	function
	\begin{equation}
	D(P)=\min_{p_{\hat{X}|Y}}\left\{\EE[d(X,\hat{X})] \;:\; d_{p}(p_{\hat{X}},p_X)\leq P\right\}.\label{APP_eq:D_P::General_definition}
	\end{equation}
	
	We can write (\ref{APP_eq:D_P::General_definition}) as 
	\begin{equation}
	D(P)=\min_{p_{\hat{X}|Y}}\left\{J(p_{\hat{X}|Y})\;:\;d_{p}(p_{\hat{X}},p_X)\leq P\right\},
	\end{equation}
	where we defined $J(p_{\hat{X}|Y})\triangleq \EE_{p_{\hat{X}Y}}[d(X,\hat{X})]$. This objective can be written as
	\begin{align}
	\label{APP_eq:J_Objective}
	J(p_{\hat{X}|Y})
	&=\EE_{p_{\hat{X}Y}}\EE[d(X,\hat{X})|Y,\hat{X}].
	\end{align}
	Let us define the cost function 
	\begin{align}
	\rho(\hat{x},y)&\triangleq \EE[d(X,\hat{X})|Y=y,\hat{X}=\hat{x}]\nonumber\\
	&=\EE[d(X,\hat{x})|Y=y],\label{eq:rho::def}
	\end{align}
	where we used the fact that $X$ is independent of $\hat{X}$ given $Y$. Then we have that the objective (\ref{APP_eq:J_Objective}) boils down to $J(p_{\hat{X}|Y})=\EE_{p_{\hat{X}Y}}\rho(\hat{X},Y)$.
	
	The problem of finding a \emph{perfect} perceptual quality estimator can be now written as an optimal transport problem
	\[
	D(P=0)=\min_{p_{\hat{X}|Y}}\EE_{p_{\hat{X}Y}}\rho(\hat{X},Y)\quad\mathrm{s.t.}\;\;p_{\hat{X}}=p_X,p_{Y}=p_Y.
	\]
	
	{
	In the setting where $\mathcal{X\mathrm{,}Y}$ are Euclidean spaces, considering the MSE distortion $d(x,\hat x) = \|x-\hat x\|^2$, we write
	\begin{align*}
		\rho(\hat x,y) & = \EEb{\|X-\hat X\|^2|Y=y,\hat X = \hat x}
		\\&
		= \EEb{\|X-\hat x\|^2 | Y=y}
		\\&
		= \EEb{\|X\|^2|Y=y} - 2\hat {x}^T\EEb{ X |Y= y} + \|\hat x \|^2
		\\&
		= \EEb{\|X-X^*\|^2|Y=y} + \left\{\EEb{\|X^*\|^2|Y=y} - 2\hat {x}^T\EEb{ X | Y=y} + \|\hat x \|^2\right\}
		\end{align*}
		
		and we have
		\begin{align*}
	 J(p_{\hat{X}|Y})=\EE_{p_{\hat X Y}} \rho(\hat X,Y) & = \EE_{p_{\hat X Y}} \EEb{\|X-X^*\|^2|Y} + \EE_{p_{\hat X Y}} \EEb{\|\hat X-X^*\|^2|Y,\hat X}
	\\&
	= D^* + \EE_{p_{\hat X Y}} \left[ \|\hat X-X^*\|^2 \right].
	\end{align*}
}

\subsection{The optimal transportation problem}
\label{appsec::Optimal Transport}
Assume $\mathcal{X\mathrm{,}Y}$ are Radon spaces \citep{ambrosio2008gradient}. Let $\rho:\mathcal{X}\times\mathcal{Y}\rightarrow \R$
be a non-negative Borel cost function, and let $q^{(x)},p^{(y)}$ be probability measures
on $\mathcal{X},\mathcal{Y}$ respectively. The optimal transport problem is then given in the following formulations.

In the \emph{Monge}
formulation, we search for an optimal transformation, often referred to as an \emph{optimal map}, $T:\mathcal{Y}\rightarrow\mathcal{X}$
minimizing 
\begin{equation}
E\rho(T(Y),Y)\,,\,\mathrm{s.t.}\,Y\sim q^{(y)},T(Y)\sim q^{(x)}.\label{APP_eq:OT::Monge_prob}
\end{equation}
Note that the Monge problem seeks for a deterministic map,
and might not have a solution.

	In the \emph{Kantorovich}
	formulation, we wish to find a probability measure $q=q_{XY}$ on $\mathcal{X}\times\mathcal{Y}$,
	minimizing 
	\begin{equation}
	E_{q}\rho(X,Y)\,,\,\mathrm{s.t.}\,q\in\Pi(q^{(x)},p^{(y)}). \label{APP_eq:OT::Kantorovich_problem}
	\end{equation}
	$\Pi$ is the set of probabilities
	on $\mathcal{X}\times\mathcal{Y}$ with marginals $q^{(x)},p^{(y)}$. {A probability minimizing \eqref{APP_eq:OT::Kantorovich_problem} is called an \emph{optimal plan}, and we denote $q\in \Pi_o(q^{(x)},p^{(y)})$.}
	Note that when $\rho(x,y)=d^p(x,y)$ and $d(x,y)$ is a metric, taking $\inf$ over
	(\ref{APP_eq:OT::Kantorovich_problem}) yields the Wasserstein distance
	$W_p^p(q^{(x)},p^{(y)})$ induced by $d(x,y)$.
	
	In the case where $\mathcal{X}=\mathcal{Y}= \R^d$ and $\rho(x,y)=\|x-y\|^2$ is the quadratic cost (and we assume $q^{(x)},p^{(y)}$ have finite first and second moments), there exists an optimal plan minimizing \eqref{APP_eq:OT::Kantorovich_problem}. If $p^{(y)}$ is absolutely continuous (w.r.t Lebesgue measure), this plan is given by an optimal map which is the unique solution to \eqref{APP_eq:OT::Monge_prob} \citep[p.5,16]{panaretos2020invitation}.

\subsection{Optimal maps between Gaussian measures} \label{appsec::Optimal Maps between Gaussian Measures}
	When $\mu_1=\mathcal{N}(m_{1},\Sigma_{1})$ and $\mu_2=\mathcal{N}(m_{2},\Sigma_{2})$ are Gaussian distributions on $\R^{d}$, we have that
\begin{equation}
W_{2}^{2}(\mu_{1},\mu_{2})=\Vert m_{1}-m_{2}\Vert_{2}^{2}+\mathrm{Tr}\left\{\Sigma_{1}+\Sigma_{2}-2\left(\Sigma_{1}^{\frac{1}{2}}\Sigma_{2}\Sigma_{1}^{\frac{1}{2}}\right)^{\frac{1}{2}}\right\}.\label{APP_eq:W_2::Gaussians}
\end{equation}

	If $\Sigma_1$ and $\Sigma_2$ are non-singular, then the distribution attaining the optimum in \eqref{eq:Wp_def} corresponds to
	\begin{equation}\label{APP_eq:optimal_dist_for_gauss}
	U\sim \mathcal{N}(m_1,\Sigma_1),\;\; V=m_{2}+T_{1\rightarrow 2}(U-m_{1}),
	\end{equation}
	where
	\begin{equation}\label{APP_eq::T1->2:def}
	T_{1\rightarrow 2}=\Sigma_{1}^{-\frac{1}{2}}\left(\Sigma_{1}^{\frac{1}{2}}\Sigma_{2}\Sigma_{1}^{\frac{1}{2}}\right)^{\frac{1}{2}}\Sigma_{1}^{-\frac{1}{2}}
	\end{equation}
	is the optimal transformation pushing forward from $\mathcal{N}(0,\Sigma_{1})$ to $\mathcal{N}(0,\Sigma_{2})$ \citep{knott1984optimal}. This transformation satisfies $\Sigma_{2}=T_{1\rightarrow 2}\Sigma_{1}T_{1\rightarrow 2}.$ 
	
	When distributions are singular, we have the following.
	\begin{lem}
		\label{Lemma::=00005BZhang,Wang.-Nehorai,-2020, Thm 3}\citep[Theorem 3]{zhang2019optimal} Let $\mu$ and $\nu$
		be two centered Gaussian measures defined on $\R^{n}$. Let $P_{\mu}$
		be the projection matrix onto $\mathrm{Im}\{\Sigma_{\mu}\}$. Then the optimal
		transport map $T_{\mu\rightarrow P_{\mu}\#\nu}$ from $\mu$ to $P_{\mu}\#\nu$
		is linear and self-adjoint, and can be written as
		\[
		T_{\mu \rightarrow P_{\mu}\#\nu}=(\Sigma_{\mu}^{1/2})^{\dagger}(\Sigma_{\mu}^{1/2}\Sigma_{\nu}\Sigma_{\mu}^{1/2})^{1/2}(\Sigma_{\mu}^{1/2})^{\dagger}.
		\]
		In the case $\mathrm{Im}\{\Sigma_{\nu}\}\subseteq \mathrm{Im}\{\Sigma_{\mu}\}$ we have
		$P_{\mu}\#\nu=\nu$, hence $T_{\mu\rightarrow \nu}=T_{\mu \rightarrow P_{\mu}\#\nu}$
		is the optimal transport map from $\mu$ to $\nu$, even where measures
		are singular. 
	\end{lem}
	
	\section{Proof of main results}
	\label{APPsec::proofs}
    
    {
    In this Section we provide proofs of the main results of this paper. In lemmas \ref{lem::orthogonality} and \ref{lem::DP w2equivalent} we present some alternative representations for $D(P)$. In Lemma \ref{lem::D-triangleinq-lowbound} we obtain a lower bound on $D(P)$. We then prove Theorem \ref{Thm::extrapol} (via a more general result given by Lemma \ref{lem:eps-interpolation}), where the lower bound of Lemma \ref{lem::D-triangleinq-lowbound} is attained. Equipped with Theorem \ref{Thm::extrapol}, we prove Theorem \ref{Thm:=00005Bthe-Distortion-Perception-funct} which is the main result of our paper. }
    	
	\subsection{Relations between $D(P)$ and $X^*$}
	\label{APP::optimalestimators}
	In this section we relate the distortion-perception function $D(P)$ given in \eqref{eq:MSE::D_P::Definition} to the estimator $X^*= \EE \left[ X|Y \right]$. Recall that  $D^*=\EEd{X}{X^*}$ and $P^*=W_2(p_X,p_{X^*}).$
	
	\begin{lem}
	\label{lem::orthogonality}
		If $\hat{X}$ is independent of $X$ given $Y$, then its MSE can be decomposed as $\EE[\|X-\hat{X}\|^2]=\EE[\|X-X^*\|^2+\EE[\|X^*-\hat{X}\|^2]$ and hence
		\begin{equation}
		D(P)=D^*+\min_{p_{\hat{X}|Y}}\left\{\EE_ {p_{\hat{X}Y}}\left[\|\hat{X}-X^{*}\|^{2} \right] \;:\; W_{2}(p_{\hat{X}},p_X) \leq P \right\}.
		\end{equation}
	\end{lem}
	\begin{proof}
		For any estimator we can write the MSE
		\begin{equation}\label{eq:app-split-error}
		\EE\left[\|X-\hat{X}\|^{2}\right]=\EE\left[\\|X-X^{*}\|^{2}\right]+\EE\left[\|\hat{X}-X^{*}\|^{2}\right]-2\EE\left[(X-X^{*})^{T}(\hat{X}-X^{*})\right].
		\end{equation}
		Since in our case $\hat{X}$ is independent of $X$ given $Y$, we show that the third term vanishes.
		\begin{align*}
		\EE\left[(X-X^{*})^{T}(\hat{X}-X^{*})\right] &= \EE\left[\EE(X-X^{*})^{T}(\hat{X}-X^{*})|Y\right]\\
				&=\EE\Big[\underbrace{\EE\left[(X-X^{*})^{T}|Y\right]}_{=0}\left[\EE(\hat{X}-X^{*})|Y\right] \Big] =0.
		\end{align*}
		Since $X^{*}$is a deterministic function of $Y$, $D^*=\EE\left[\|X-X^{*}\|^{2}\right]$
		is a property of the problem, and does not depend on the choice of
		$p_{\hat{X}|Y}$, which, in view of \eqref{eq:app-split-error} completes the proof.
		\end{proof}
Next, we express $D(P)$ in terms of the Wasserstein distance between $p_{\hat X}$ and $p_{X^*}$.	
\begin{lem}[Eq.~(14)]
	\label{lem::DP w2equivalent}
	\begin{equation}
	D(P)=D^*+\min_{p_{\hat{X}}}\left \{ W^2_2(p_{\hat{X}},p_{X^{*}}) \;:\; W_{2}(p_{\hat{X}},p_X)\leq P \right \}.
	\end{equation}
\end{lem}
\begin{proof}
	Denote $W_2^2(\mathcal{B}_P,p_{X^{*}})=\min_{p_{\hat{X}}:W_{2}(p_{\hat{X}},p_X)\leq P}W^2_2(p_{\hat{X}},p_{X^{*}})$, where $\mathcal{B}_P$ is the ball of radius $P$ around $p_X$ in Wasserstein space.
	
	From Lemma \ref{lem::orthogonality} we have
	\begin{equation}
	D(P)=D^*+\min_{p_{\hat{X}|Y}:W_{2}(p_{\hat{X}},p_X)\leq P}\EE_ {p_{\hat{X}Y}}\left[\|\hat{X}-X^{*}\|^{2} \right].
	\end{equation}
	For every $p_{\hat X |Y}$ whose marginal attains $W_{2}(p_{\hat{X}},p_X)\leq P$ we have,
	\begin{align*}
	\EE_{p_{\hat{X}Y}} \left[\|\hat{X}-X^{*}\|^{2} \right]  & \geq \inf_{q\in\Pi(p_{\hat{X}},p_{X^{*}})} \EE_q\left[\|\hat{X}-X^{*}\|^{2} \right] \\
	&= W_2^2(p_{\hat{X}},p_{X^{*}}) \\
	& \geq \min_{p_{\hat{X}}:W_{2}(p_{\hat{X}},p_X)\leq P}W^2_2(p_{\hat{X}},p_{X^{*}}),
	\end{align*} 
	which leads to $D(P) \geq D^* + W_2^2(\mathcal{B}_P,p_{X^{*}})$.
	
	Conversely, given $p_{\hat{X}}$ such that $W_{2}(p_{\hat{X}},p_X)\leq P$, we have an optimal plan $p_{\hat X{X^*}}$ achieving $W_{2}(p_{\hat{X}},p_{X^*})$. Once we determine the optimal plan $p_{\hat X {X^*}}$ with marginal $p_{\hat X}$, we have an estimator $\hat X$ given by $p_{\hat X| Y}$ achieving $\EE_{p_{\hat{X} Y}}\left[\|\hat{X}-X^{*}\|^{2} \right] = W^2_{2}(p_{\hat{X}},p_{X^*})$ (for the connection between the optimal plan $p_{\hat X X^*}$ and the choice of a consistent $p_{\hat X|Y}$, see Remark about uniqueness in Sec.~\ref{sec::The MSE--Wasserstein-2 tradeoff}). We then have 
	\begin{align*}
	\min_{p_{\hat{X}|Y}:W_{2}(p_{\hat{X}},p_X)\leq P}\EE_ {p_{\hat{X}Y}}\left[\|\hat{X}-X^{*}\|^{2} \right] & \leq \EE_{p_{\hat{X}Y}} \left[\|\hat{X}-X^{*}\|^{2} \right] = W^2_{2}(p_{\hat{X}},p_{X^*}). 
	\end{align*}
	Taking the minimum over $p_{\hat{X}}$ yields $D(P)\leq D^*+ W_2^2(\mathcal{B}_P,p_{X^{*}})$. Combining the upper and lower bounds, we obtain the desired result.
\end{proof}

For the proof of Theorem~3, we first prove the following
\begin{lem} \label{lem::D-triangleinq-lowbound}
	$D(P)\geq D^*+[(P^*-P)_+]^{2}$.
\end{lem}
\begin{proof}
	For every estimator satisfying $W_{2}(p_{\hat{X}},p_X)\leq P$,
	we have from the triangle inequality 
	\begin{equation}\label{eq:app-P*-P}
	P^*=W_{2}(p_X,p_{X^{*}})\leq W_{2}(p_{\hat{X}},p_{X^{*}})+W_{2}(p_{\hat{X}},p_X)\leq W_{2}(p_{\hat{X}},p_{X^{*}})+P,
	\end{equation}
	yielding 
	\begin{align*}
	\EE\left[\|X-\hat{X}\|^{2}\right] &= \EEd{X}{X^*} + \EEd{\hat{X}}{X^*} \\
		&\geq D^*+W_{2}^{2}(p_{\hat{X}},p_{X^{*}})\\
		&\geq D^*+(P^*-P)_+^{2},
	\end{align*}
	where the last inequality follows from \eqref{eq:app-P*-P}. 
	Hence $D(P)=\min_{p_{\hat{X}|Y}:W_{2}(p_{\hat{X}},p_X)\leq P}\EE_{p_{\hat{X}Y}}\left[\|X-\hat{X}\|^{2}\right]\geq D^*+[(P^*-P)_+]^{2}$.
\end{proof}

\subsubsection{Proof of Theorem \ref{Thm::extrapol}}
\label{APPsec:: Proof of Theorem 3}

	\begin{thm*}\it{\bf{\ref{Thm::extrapol}}}. Let $\hat{X}_0$ be an estimator achieving perception index~$0$ and MSE $D(0)$. Then for any $P\in[0,P^*]$,
		the estimator
		\begin{equation}
		\hat{X}_P=\left(1-\frac{P}{P^*}\right)\hat{X}_{0}+\frac{P}{P^*}X^{*}\label{App_eq:hat_X::extrapolated}
		\end{equation}
		is optimal for perception index $P$, namely, it achieves perception index $P$ and distortion $D(P)$.
	\end{thm*}

Let us prove a stronger result, from which Theorem \ref{Thm::extrapol} will follow.
\begin{lem} \label{lem:eps-interpolation}
	Let $\hat{X_{\varepsilon}}$ be an estimator (independent of $X$
	given $Y$) achieving $W_{2}(p_X,p_{\hat{X}_{\varepsilon}})\leq\varepsilon_{P}$
	and $\EE\left[\|\hat{X_{\varepsilon}}-X^{*}\|^{2}\right]\leq(1+\varepsilon_{D})^{2}W_{2}^{2}(p_X,p_{X^{*}})$
	for some $\varepsilon_{D},\varepsilon_{P}\geq0$. Given $0\leq P\leq P^*=W_{2}(p_X,p_{X^{*}})$,
	consider the estimator
	\begin{equation}
	\hat{X}_P=\left(1-\frac{P}{ P^*}\right)\hat{X_{\varepsilon}}+\frac{P}{ P^*}X^{*}.\label{eq:hat_X::extrapolated-1-1}
	\end{equation}
	Then $\hat{X}_P$ achieves $\EE[\|X-\hat{X}_P\|^{2}]\leq D^*+(1+\varepsilon_{D})^{2}(P^*-P)^{2}$
	with perception index $\varepsilon_{P}+(1+\varepsilon_{D})P$. When $\varepsilon_{D},\varepsilon_{P}=0$, namely $\hat{X}_{\varepsilon}$ is
	an optimal perfect perceptual quality estimator, $\hat{X}_P$ is
	an optimal estimator under perception constraint $P$, which proves Theorem~3.
\end{lem}
\begin{proof}
	$W_{2}^{2}(p_{\hat{X}_{\varepsilon}},p_{\hat{X}_P})\leq \EE\left[\|\hat{X_{\varepsilon}}-\hat{X}_P\|^{2}\right]$,
	and using the triangle inequality
	\begin{align*}
		W_{2}(p_X,p_{\hat{X}_P}) &\leq  W_{2}(p_X,p_{\hat{X}_{\varepsilon}})+W_{2}(p_{\hat{X}_{\varepsilon}},p_{\hat{X}_P})\\
		&\leq\varepsilon_{P}+\sqrt{\EE\left[\|\hat{X_{\varepsilon}}-\hat{X}_P\|^{2}\right]}\\
		&=\varepsilon_{P}+\sqrt{\frac{P^{2}}{W_{2}^{2}(p_X,p_{X^{*}})}\EE\left[\|\hat{X_{\varepsilon}}-X^{*}\|^{2}\right]}\\
		&\leq\varepsilon_{P}+P(1+\varepsilon_{D}), 
	\end{align*}
	where the equality is based on \eqref{eq:hat_X::extrapolated-1-1}.
	A direct calculation of the distortion yields
	\begin{align*}
	\EE\left[\|X^{*}-\hat{X}_P\|^{2}\right]&=\left(1-\frac{P}{W_{2}(p_X,p_{X^{*}})}\right)^{2}
	     \EE\left[\|X^{*}-\hat{X_{\varepsilon}}\|^{2}\right]\\
	     &\leq(1+\varepsilon_{D})^{2}(W_{2}(p_X,p_{X^{*}})-P)^{2},\\
		\EE\left[\|X-\hat{X}_P\|^{2}\right]
		&=D^*+\EE\left[\|X^{*}-\hat{X}_P\|^{2}\right]\\
		&\leq D^*+(1+\varepsilon_{D})^{2}(W_{2}(p_X,p_{X^{*}})-P)^{2}.
	\end{align*}
	When $\varepsilon_{D},\varepsilon_{P}=0$ we have $W_{2}(p_X,p_{\hat{X}_P})\leq P$
	and $\EE\left[\|X-\hat{X}_P\|^{2}\right]\leq D^*+(W_{2}(p_X,p_{X^{*}})-P)^{2}$.
	From Lemma \ref{lem::D-triangleinq-lowbound}, the latter inequality is achieved with equality. Note that since here $\EEd{\hat X_\varepsilon}{X^*}=W^2_2(p_X,p_{X^*})$, the distributions of $\{ \hat X_P, \ P \in \left[ 0,W_2(p_X,p_{X^*}) \right] \}$ form a constant-speed geodesic, hence $W_2(p_X,p_{\hat{X}_P}) = P$.
\end{proof}

	\begin{cor}
		\label{cor::hatx deterministic if X* has density}  When ${X^*}$
		has a density, $\hat X_0$ (hence $\hat{X}_P$) can be obtained via a deterministic transformation of $Y$.
	\end{cor}
\begin{proof}
	Since the distribution of $X^*$ is absolutely continuous, we have an optimal map $T_{p_{X^*}\rightarrow p_X }$ between the distributions of $X^*$ and $X$ (see discussion in App.~\ref{appsec::Optimal Transport}). Namely, we have that $\hat X_0 = T_{p_{X^*}\rightarrow p_X }(X^*)$ is an optimal estimator with perception index $0$. Thus, according to \eqref{eq:hat_X::extrapolated} $\hat X_P =\left(1-\frac{P}{P^*}\right)T_{p_{X^*}\rightarrow p_X }(X^*) + \frac{P}{P^*}X^* $ are optimal estimators, which in this case are given by a deterministic function of $Y$.
\end{proof}

	\subsection{Proof of theorem \ref{Thm:=00005Bthe-Distortion-Perception-funct}}
	
	With Theorem \ref{Thm::extrapol} and Lemma \ref{lem:eps-interpolation} in hand, we are now ready to prove our main result.
	
	\begin{thm*} \it{\bf{\ref{Thm:=00005Bthe-Distortion-Perception-funct}}}. The DP function \eqref{eq:MSE::D_P::Definition} is given by
	\begin{equation}
	D(P)=
	D^*+\left[(P^*-P)_{+}\right]^{2}.
	\label{APP_eq:Thm:DP_function::Main_result}
	\end{equation}
	Furthermore, an estimator achieving perception index $P$ and distortion $D(P)$ can always be constructed by applying a (possibly stochastic) transformation to $X^*$.
\end{thm*}

\begin{proof}
When $P\geq P^*$ the result is trivial since $D(P)=D^*$. Let us focus on  $P<P^*$. Since $X,X^*\in \R^{n_x}$, we have an optimal plan $p_{\hat X_0 X^*}$ between their distributions, attaining $P^*$ \citep{ambrosio2008gradient,panaretos2020invitation}. We then have an optimal estimator $\hat X_0$ with perception index $0$, which is given by this joint distribution hence achieving $\EEd{\hat X_0}{X^*}=(P^*)^2$ (for the connection between $p_{\hat X_0 X^*}$ and the choice of $p_{\hat X_0|Y}$, see Remark about uniqueness in Sec.~\ref{sec::The MSE--Wasserstein-2 tradeoff}). For any perception $P<P^*$, consider $\hat X_P$ given by \eqref{App_eq:hat_X::extrapolated}. We have $W_2(p_X,p_{\hat X_P})=P$, and (see Theorem \ref{Thm::extrapol}'s proof)
\[\EE\left[\|X-\hat{X}_P\|^{2}\right]\leq D^*+(W_{2}(p_X,p_{X^{*}})-P)^{2},\]
hence $D(P)\leq	D^*+\left[(P^*-P)_{+}\right]^{2}$. On the other hand, we have (Lemma \ref{lem::D-triangleinq-lowbound}) $D(P)\geq 	D^*+\left[(P^*-P)_{+}\right]^{2}$, which completes the proof.
\end{proof}

	\subsection{The Gaussian setting}
\label{appsec::GaussianSetting}	
In this Section we prove Theorems \ref{thm:Gaussian1} and \ref{thm:Gaussian_not_unique}. We begin by proving Theorem \ref{thm:Gaussian_not_unique}, and then show that Theorem \ref{thm:Gaussian1} follows as a special case.
Recall that 
\begin{equation}
(G^*)^2 = \tr{\Sigma_X + \Sigma_{X^*} -2\left(\Sigma_X^{1/2}\Sigma_{X^*}\Sigma_X^{1/2}\right)^{1/2}} \end{equation} 
and
\begin{equation} T^{*}=\Sigma_{X}^{-1/2}(\Sigma_{X}^{1/2}\Sigma_{X^*}\Sigma_{X}^{1/2})^{1/2}\Sigma_{X}^{-1/2}. 
\end{equation}

\begin{thm*} \it{\bf{\ref{thm:Gaussian_not_unique}.}}
	Consider the setting of Theorem~\ref{thm:Gaussian1} in the main text. Let $\Sigma_{\hat X_0 Y}\in \R^{n_{x}\times n_{y}}$ satisfy
		\begin{equation}
		\Sigma_{\hat X_0 Y}\Sigma_{Y}^{-1}\Sigma_{YX}=\Sigma_{X}^{\frac{1}{2}}(\Sigma_{X}^{\frac{1}{2}}\Sigma_{X^*}\Sigma_{X}^{\frac{1}{2}})^{\frac{1}{2}}\Sigma_{X}^{-\frac{1}{2}},\label{APP_eq:Thm_Gaussian_General::M_cond1-2-1-1}
		\end{equation}
		and $W_{0}$ be a zero-mean Gaussian noise with covariance
		\begin{equation}
		\Sigma_{W_{0}}=\Sigma_{X}-\Sigma_{\hat X_0 Y}\Sigma_{Y}^{-1}\Sigma_{\hat X_0 Y}^{T}\succeq0
		\label{APP_eq:Thm_Gaussian_General::M_cond2-2-1-1}
		\end{equation}
		that is independent of $Y,X$. Then, for any $P\in[0,G^*]$, an optimal estimator with perception index $P$ can be obtained by
		\begin{equation}
		\label{APP_eq::X_P_Gauss:interpolation_not_unique}
		\hat{X}_P=\left(\left(1-\frac{P}{G^{*}}\right)\Sigma_{\hat X_0 Y}+\frac{P}{G^{*}}\Sigma_{XY}\right)\Sigma_{Y}^{-1}Y+\left(1-\frac{P}{G^{*}}\right)W_{0}.
		\end{equation} The estimator given in \eqref{APP_eq::X_P_Gauss:interpolation}
		is one solution to \eqref{APP_eq:Thm_Gaussian_General::M_cond1-2-1-1}-\eqref{APP_eq:Thm_Gaussian_General::M_cond2-2-1-1},
		but it is generally not unique.\end{thm*}

\begin{proof} (Theorem \ref{thm:Gaussian_not_unique})
	Let $\hat X_0 \triangleq \Sigma_{\hat X_0 Y} \Sigma_{Y}^{-1}Y + W_0$ where $\Sigma_{\hat X_0 Y}$ satisfies \eqref{APP_eq:Thm_Gaussian_General::M_cond1-2-1-1}-\eqref{APP_eq:Thm_Gaussian_General::M_cond2-2-1-1}. It is easy to see that $\hat X_0 \sim \mathcal{N}(0,\Sigma_{X})$ and it is jointly Gaussian with $(X,Y,X^*)$. We have by \eqref{APP_eq:Thm_Gaussian_General::M_cond1-2-1-1}
	\begin{equation}\label{APP:x0x*corr}
	    	\EEb{X^*\hat X_0^T}=\Sigma_{XY}\Sigma_{Y}^{-1}\Sigma_{Y \hat X_0} = \Sigma_{X}^{-1/2}(\Sigma_{X}^{1/2}\Sigma_{X^*}\Sigma_{X}^{1/2})^{1/2}\Sigma_{X}^{1/2}, \end{equation}
	hence using \eqref{APP_eq:Thm_Gaussian_General::M_cond2-2-1-1},
	\begin{align*}
	\EEb{\|\hat X_0-X^*\|^2}&=\tr{\Sigma_X + \Sigma_{X^*}-2\EEb{X^*\hat X_0^T}}\\
	&=\tr{\Sigma_X + \Sigma_{X^*}-2\Sigma_{X}^{-1/2}(\Sigma_{X}^{1/2}\Sigma_{X^*}\Sigma_{X}^{1/2})^{1/2}\Sigma_{X}^{1/2}}\\
	&=\tr{\Sigma_X + \Sigma_{X^*}-2(\Sigma_{X}^{1/2}\Sigma_{X^*}\Sigma_{X}^{1/2})^{1/2}}\\
	&=G^2(\Sigma_{X},\Sigma_{X^*})\\
	&=(G^*)^2.
	\end{align*}
	Summarizing, $\hat X_0$ is an optimal perfect perception quality estimator. Note that \eqref{APP_eq::X_P_Gauss:interpolation_not_unique} can be written as
	\[ \hat X_P = \left(1- \frac{P}{G^*} \right)\hat X_0 + \frac{P}{G^*}X^*,\] and by Theorem \ref{Thm::extrapol} we have that it is an optimal estimator.
	\end{proof}

	Before proceeding to the proof of Theorem \ref{thm:Gaussian1}, let us introduce some auxiliary facts.

	\begin{lem}
		\label{Lemma:: ker sigma* in ker Tp*} Let $\Sigma,\Sigma_{X^*}\in \R^{n\times n}$
		be (symmetric) PSD matrices, and $\Sigma_{X}\in \R^{n\times n}$ is
		PD. Denote $T^*=\Sigma_{X}^{-\frac{1}{2}}\left(\Sigma_{X}^{\frac{1}{2}}\Sigma_{X^*}\Sigma_{X}^{\frac{1}{2}}\right)^{\frac{1}{2}}\Sigma_{X}^{-\frac{1}{2}}$. Then: 
	\end{lem}
	\begin{enumerate}
		\item $\Ker\{\Sigma\}=\Ker\{\Sigma^{\frac{1}{2}}\}$.
		\item $\Ker\{\Sigma_{*}\}\subseteq \Ker\{\Sigma_{X}^{\frac{1}{2}}(\Sigma_{X}^{\frac{1}{2}}\Sigma_{X^*}\Sigma_{X}^{\frac{1}{2}})^{\frac{1}{2}}\Sigma_{X}^{-\frac{1}{2}}\}=\Ker\{\Sigma_{X}T^{*}\}$, and
		we have $\Sigma_{X}T^{*}\Sigma_{X^*}^{\dagger}\Sigma_{X^*}=\Sigma_{X}T^{*}$.
	\end{enumerate}
	\begin{proof}
		(1) Let $\Sigma$ be PSD. Since it is real and symmetric it is diagonalizable,
		$\Sigma=UDU^{T}$ and $\Sigma^{1/2}=UD^{1/2}U^{T}$ where $D$ is
		a diagonal matrix with non-negative entries which are the eigenvalues
		of $\Sigma$. We have $\Ker\{D\}=\Ker\{D^{1/2}\}=\{v\in \R^{n}:v_{i}=0\,\forall i:D_{i,i}\neq0\}$
		and since $U$ is full-rank, $\Ker\{\Sigma\}=\Ker\{\Sigma^{1/2}\}=U\Ker\{D\}$.
		
		(2) Assume $\Sigma_{X^*}v=0$. We have $(\Sigma_{X}^{1/2}\Sigma_{X^*}\Sigma_{X}^{1/2})\Sigma_{X}^{-1/2}v=0$,
		implying that $\Sigma_{X}^{-1/2}v\in \Ker\{(\Sigma_{X}^{1/2}\Sigma_{X^*}\Sigma_{X}^{1/2})\}=\Ker\{(\Sigma_{X}^{1/2}\Sigma_{X^*}\Sigma_{X}^{1/2})^{1/2}\}.$
		The equality is true since $\Sigma_{X}^{1/2}\Sigma_{X^*}\Sigma_{X}^{1/2}=\Sigma_{X}^{1/2}\Sigma_{X^*}^{1/2}(\Sigma_{X}^{1/2}\Sigma_{X^*}^{1/2})^{T}$
		is PSD, and we use (1). To conclude, we have 
		\[
		\Sigma_{X}T^{*}v=\Sigma_{X}^{1/2}(\Sigma_{X}^{1/2}\Sigma_{X^*}\Sigma_{X}^{1/2})^{1/2}\Sigma_{X}^{-1/2}v=0\Longrightarrow \Ker\{\Sigma_{X^*}\}\subseteq \Ker\{\Sigma_{X}T^{*}\}.
		\]
		Recall now that $(I-\Sigma_{X^*}^{\dagger}\Sigma_{X^*})$ is a projection
		onto $\Ker\{\Sigma_{X^*}\}$. We have $\Sigma_{X}T^{*}(I-\Sigma_{X^*}^{\dagger}\Sigma_{X^*})=0$,
		yielding $\Sigma_{X}T^{*}\Sigma_{X^*}^{\dagger}\Sigma_{X^*}=\Sigma_{X}T^{*}$.
	\end{proof}
	
	The following Lemma is a reminder of Schur's Complement
	and its properties.
	\begin{lem}
		{[}Schur's complement{]}. Let $\Sigma=\left[\begin{array}{cc}
		A & B\\
		B^{T} & C
		\end{array}\right]$ be a symmetric matrix where $A$ is PD. Then $\nicefrac{\Sigma}{A}\triangleq C-B^{T}A^{-1}B$
		is the Schur complement of $\Sigma$, and we have that $\Sigma$ is
		PSD iff $\nicefrac{\Sigma}{A}$ is PSD.
	\end{lem}

    We are now ready to prove Theorem~4.
	\begin{thm*} \it{\bf{\ref{thm:Gaussian1}.}}
	Assume $X$ and $Y$ are zero-mean jointly Gaussian random vectors with $\Sigma_X,\Sigma_{Y} \succ0$.
	Then for any $P\in[0,G^{*}]$, an estimator with perception index $P$ and MSE $D(P)$ can be constructed as
	\begin{equation} \label{APP_eq::X_P_Gauss:interpolation}
	\hat{X}_P=\left(\left(1-\frac{P}{G^{*}}\right)\Sigma_{X}^{\frac{1}{2}}\left(\Sigma_{X}^{\frac{1}{2}}\Sigma_{X^*}\Sigma_{X}^{\frac{1}{2}}\right)^{\frac{1}{2}}\Sigma_{X}^{-\frac{1}{2}}\Sigma_{X^*}^{\dagger}+\frac{P}{G^{*}}I\right)\Sigma_{XY}\Sigma_{Y}^{-1}Y+\left(1-\frac{P}{G^{*}}\right)W,
	\end{equation}
	where $W$ is a zero-mean Gaussian noise with covariance $\Sigma_{W}=\Sigma_{X}^{1/2}(I-\Sigma_{X}^{1/2}T^{*}\Sigma_{X^*}^{\dagger}T^{*}\Sigma_{X}^{1/2})\Sigma_{X}^{1/2}$, which is independent of $Y,X$.
\end{thm*}

	\begin{proof} 
	We observe that \eqref{APP_eq::X_P_Gauss:interpolation} is a special case of \eqref{APP_eq::X_P_Gauss:interpolation_not_unique}, where $\Sigma_{\hat X_0 Y}=\Sigma_{ Y \hat   X_0}^T = \Sigma_{X}^{\frac{1}{2}}\left(\Sigma_{X}^{\frac{1}{2}}\Sigma_{X^*}\Sigma_{X}^{\frac{1}{2}}\right)^{\frac{1}{2}}\Sigma_{X}^{-\frac{1}{2}}\Sigma_{X^*}^{\dagger}\Sigma_{XY}$. We now show that $\Sigma_{\hat X_0 Y}$ has the desired properties \eqref{APP_eq:Thm_Gaussian_General::M_cond1-2-1-1}-\eqref{APP_eq:Thm_Gaussian_General::M_cond2-2-1-1}. By substitution,
	\begin{align*}
	\Sigma_{\hat X_0 Y}\Sigma_{Y}^{-1}\Sigma_{YX} & = \Sigma_{X}^{\frac{1}{2}}\left(\Sigma_{X}^{\frac{1}{2}}\Sigma_{X^*}\Sigma_{X}^{\frac{1}{2}}\right)^{\frac{1}{2}}\Sigma_{X}^{-\frac{1}{2}}\Sigma_{X^*}^{\dagger}\left( \Sigma_{XY} \Sigma_{Y}^{-1}\Sigma_{YX} \right)
	\\ & = \Sigma_{X}^{\frac{1}{2}}\left(\Sigma_{X}^{\frac{1}{2}}\Sigma_{X^*}\Sigma_{X}^{\frac{1}{2}}\right)^{\frac{1}{2}}\Sigma_{X}^{-\frac{1}{2}}\Sigma_{X^*}^{\dagger}\Sigma_{X^*} \\ & =\Sigma_{X}^{\frac{1}{2}}\left(\Sigma_{X}^{\frac{1}{2}}\Sigma_{X^*}\Sigma_{X}^{\frac{1}{2}}\right)^{\frac{1}{2}}\Sigma_{X}^{-\frac{1}{2}}.		
	\end{align*}
	The last equality is due to Lemma \ref{Lemma:: ker sigma* in ker Tp*}.
	
	Recall $\Sigma_{X^*}^{\dagger}\Sigma_{X^*}\Sigma_{X^*}^{\dagger}=\Sigma_{X^*}^{\dagger}$, and we denote $T^*=\Sigma_{X}^{-\frac{1}{2}}\left(\Sigma_{X}^{\frac{1}{2}}\Sigma_{X^*}\Sigma_{X}^{\frac{1}{2}}\right)^{\frac{1}{2}}\Sigma_{X}^{-\frac{1}{2}}$. We now have
	\begin{align*}
	\Sigma_{Y\hat X_0}\Sigma_{X}^{-1}\Sigma_{ \hat X_0 Y} & = \Sigma_{YX}\Sigma_{X^*}^{\dagger}T^*\Sigma_{X}\Sigma_X^{-1} \Sigma_{X} T^*\Sigma_{X^*}^{\dagger} \Sigma_{XY} 
	\\ & = \Sigma_{YX}\Sigma_{X^*}^{\dagger}\Sigma_{X}^{-\frac{1}{2}} (\Sigma_{X}^{\frac{1}{2}} \Sigma_{X^*} \Sigma_{X}^{\frac{1}{2}})\Sigma_{X}^{-\frac{1}{2}} \Sigma_{X^*}^{\dagger} \Sigma_{XY} 
	\\ &= 
	 \Sigma_{YX}\Sigma_{X^*}^{\dagger} \Sigma_{X^*}  \Sigma_{X^*}^{\dagger} \Sigma_{XY}
	\\ &=
	\Sigma_{YX}\Sigma_{X^*}^{\dagger}\Sigma_{XY},
	\end{align*}
	hence
	\begin{equation} \label{APP_eq:: Sw schur}
	\Sigma_Y -\Sigma_{Y\hat X_0}\Sigma_{X}^{-1}\Sigma_{ \hat X_0 Y} =  \Sigma_Y -\Sigma_{YX}\Sigma_{X^*}^{\dagger}\Sigma_{XY} = \Sigma_{Y|X^*} \succeq 0.
	\end{equation} 
	Since $\Sigma_X, \Sigma_Y \succ 0$, \eqref{APP_eq:: Sw schur} is Schur's complement of $\begin{bmatrix} \Sigma_X &\Sigma_{ \hat X_0 Y}\\ \Sigma_{Y \hat X_0 }&\Sigma_Y \end{bmatrix} \succeq 0 $, yielding
	\begin{equation} \label{APP_eq:: Sw PSD}
	\Sigma_{W}=\Sigma_{X}-\Sigma_{\hat X_0 Y}\Sigma_{Y}^{-1}\Sigma_{\hat X_0 Y}^{T}\succeq 0.
	\end{equation} 
\end{proof}

	\begin{cor}[Non-singular special case]
		\label{cor::sigma* invertible}  In the case where $\Sigma_{X^*}$
			is invertible, $\Sigma_{\hat X_0 Y}=\Sigma_{X}T^{*}\Sigma_{X^*}^{-1}\Sigma_{XY}$ in the proof of Theorem \ref{thm:Gaussian1},
			and it is easy to see that the noise covariance is $\Sigma_{W}=0$. In this case $\Sigma_{\hat X_0 Y}$ is the unique solution to \eqref{APP_eq:Thm_Gaussian_General::M_cond1-2-1-1}-\eqref{APP_eq:Thm_Gaussian_General::M_cond2-2-1-1}. This means that $\hat{X_0}$ (hence $\hat{X_P}$) is a deterministic function of $Y$.
	\end{cor}
	\begin{proof}
	   		We first show $\Sigma_{W}=0$. Let $M_P=\Sigma_{\hat X_0 Y}=\Sigma_{X}T^{*}\Sigma_{X^*}^{-1}\Sigma_{XY}$, then
		\begin{align*}
		\Sigma_{W} &=\Sigma_{X}-M_{P}\Sigma_{Y}^{-1}M_{P}^{T}
\\ &
		= \Sigma_{X}-\Sigma_{X}T^{*}\Sigma_{X^*}^{-1}\Sigma_{XY}\Sigma_{Y}^{-1}\Sigma_{YX}\Sigma_{X^*}^{-1}T^{*}\Sigma_{X}
	\\ & 
		=\Sigma_{X}-\Sigma_{X}\Sigma_{X}^{-1/2}(\Sigma_{X}^{1/2}\Sigma_{X^*}\Sigma_{X}^{1/2})^{1/2}\underbrace{(\Sigma_{X}^{-1/2}\Sigma_{X^*}^{-1}\Sigma_{X}^{-1/2})}_{=(\Sigma_{X}^{1/2}\Sigma_{X^*}\Sigma_{X}^{1/2})^{-1}}(\Sigma_{X}^{1/2}\Sigma_{X^*}\Sigma_{X}^{1/2})^{1/2}\Sigma_{X}^{-1/2}\Sigma_{X}
		 \\& =\Sigma_{X}-\Sigma_{X}\Sigma_{X}^{-1/2}\Sigma_{X}^{-1/2}\Sigma_{X}=0.
		 \end{align*}
		Now, assume $M$ is a solution to \eqref{APP_eq:Thm_Gaussian_General::M_cond1-2-1-1}-\eqref{APP_eq:Thm_Gaussian_General::M_cond2-2-1-1},
		then $M_{\Delta}=M_{P}-M$ satisfies $M_{\Delta}\Sigma_{Y}^{-1}\Sigma_{YX}=0$
		and
		\begin{flalign*}
		&\Sigma_{X}-M\Sigma_{Y}^{-1}M^{T} =& \\ &\Sigma_{X}-[M_{P}\Sigma_{Y}^{-1}M_{P}^{T}+M_{\Delta}\Sigma_{Y}^{-1}M_{\Delta}^{T}-M_{\Delta}\Sigma_{Y}^{-1}M_{P}^{T}-M_{P}\Sigma_{Y}^{-1}M_{\Delta}^{T}]\succeq0.&
		\end{flalign*}
		But, $M_{\Delta}\Sigma_{Y}^{-1}M_{P}^{T}=(M_{\Delta}\Sigma_{Y}^{-1}\Sigma_{YX})\Sigma_{X^*}^{-1}T^{*}\Sigma_{X}=0$
		and $\Sigma_{X}-M_{P}\Sigma_{Y}^{-1}M_{P}^{T}=0$, yielding $M_{\Delta}\Sigma_{Y}^{-1}M_{\Delta}^{T}\preceq0$.
		Since $M_{\Delta}\Sigma_{Y}^{-1}M_{\Delta}^{T}$
		is PSD and $\Sigma_{Y}^{-1}$ is PD, we conclude that $M_{\Delta}=0$.
	\end{proof}

	\section{Settings with commuting covariances}
\label{appsec::commute}
In many practical problems, covariance matrices may have the commutative relation $\Sigma_{X}\Sigma_{X^*}=\Sigma_{X^*}\Sigma_{X}$. This is the case, for example, of circulant or large Toeplitz matrices \citep{gray2006toeplitz}. For natural images this is a reasonable assumption since shift-invariance induces diagonalization by the Fourier basis \citep{unser1984approximation}. 

In the Gaussian settings of Sec.~\ref{Sec::GaussianSetting}, where $\Sigma_{X},\Sigma_{X^*}$ commute it is easy to see that the Gelbrich distance between them can be written as \[G^*=G((\mu_X,\Sigma_{X}),(\mu_{X^*},\Sigma_{X^*}))=\|\Sigma_{X}^{1/2}-\Sigma_{X^*}^{1/2} \|_F.\] $\|A\|_F = \sqrt{\tr{A^TA}}$ is the Frobenius norm. This is due to the fact that $\Sigma_{X}^{1/2},\Sigma_{X^*}^{1/2}$ also commute. In order to achieve $\EEb{\| \hat X_0 - X^* \|^2}=(G^*)^2$, an optimal perfect perception quality estimator has to satisfy  \eqref{APP:x0x*corr} which now takes the form 
\[\EEb{X^* \hat X_0^T} = \Sigma_{X}^{1/2}\Sigma_{X^*}^{1/2}.\] 
It is easy to see that estimators obtained by $\hat X_0 ,X^*$ using \eqref{eq:hat_X::extrapolated} are Gaussian with zero mean and covariance $\Sigma_P$, given by
\begin{equation}
\Sigma_P^{\frac{1}{2}}=\left(1-\frac{P}{G^*}\right)\Sigma_X^{\frac{1}{2}} +\frac{P}{G^*}\Sigma_{X^*}^{\frac{1}{2}}.
\end{equation}
Pay attention that since the roots commute, $\Sigma_P$ commmutes with $\Sigma_{X},\Sigma_{X^*}$, and
		\[\|\Sigma_X^{\frac{1}{2}} - \Sigma_P^{\frac{1}{2}}\|_F = P, \quad \|\Sigma_P^{\frac{1}{2}} - \Sigma_{X^*}^{\frac{1}{2}}\|_F = G^*-P. \]
		
		This further reduces the geometry of the problem to the $l^{2}$-distance between commuting matrices.

\section{Numerical illustration}
\subsection{Simulation details}

For each algorithm, we acquire $100$ RGB images which are reconstructions of BSD100 images. We extract $9 \times 9$ patches from the RGB images, and then estimate:
\[ m_\text{Alg} = \frac{1}{N_\text{patches}} \sum_i p_i, \quad \Sigma_\text{Alg} =\frac{1}{N_\text{patches}-1} (p_i -m_\text{Alg})(p_i-m_\text{Alg})^T, \]
where $p_i$ is the $i$-th patch (a $243$-row vector, $N_\text{patches}=1,632,800$).
We compute using \eqref{eq:Gelbrich_dist}
\[ \mathrm{MSE}_\text{Alg} = \frac{1}{243 \times N_\text{patches}} \sum_i \|p^{\text{Alg}}_i - p^{\text{BSD100}}_i\|^2, \quad P_\text{Alg} = \sqrt{\frac{1}{243}} G\left((m_\text{BSD100},\Sigma_\text{BSD100}) ,(m_\text{Alg},\Sigma_\text{Alg}) \right).\]

The estimators $\hat {X}_t$ are constructed using per-pixel interpolation between EDSR and ESRGAN 
	\[ \hat X_t = tX_{\text{EDSR}} + (1-t)X_{\text{ESRGAN}}.\]

\subsection{Visual illustration}

Here we present a visual comparison between SR methods and our constructed estimators, achieving roughly the same MSE but with a lower perception index. We also present EDSR, ESRGAN, the low-resolution input, and the ground-truth BSD100 images.

\begin{figure}
		\begin{centering}
			\includegraphics[bb=250bp 23bp 850bp 750bp,clip,width = 0.95\linewidth]{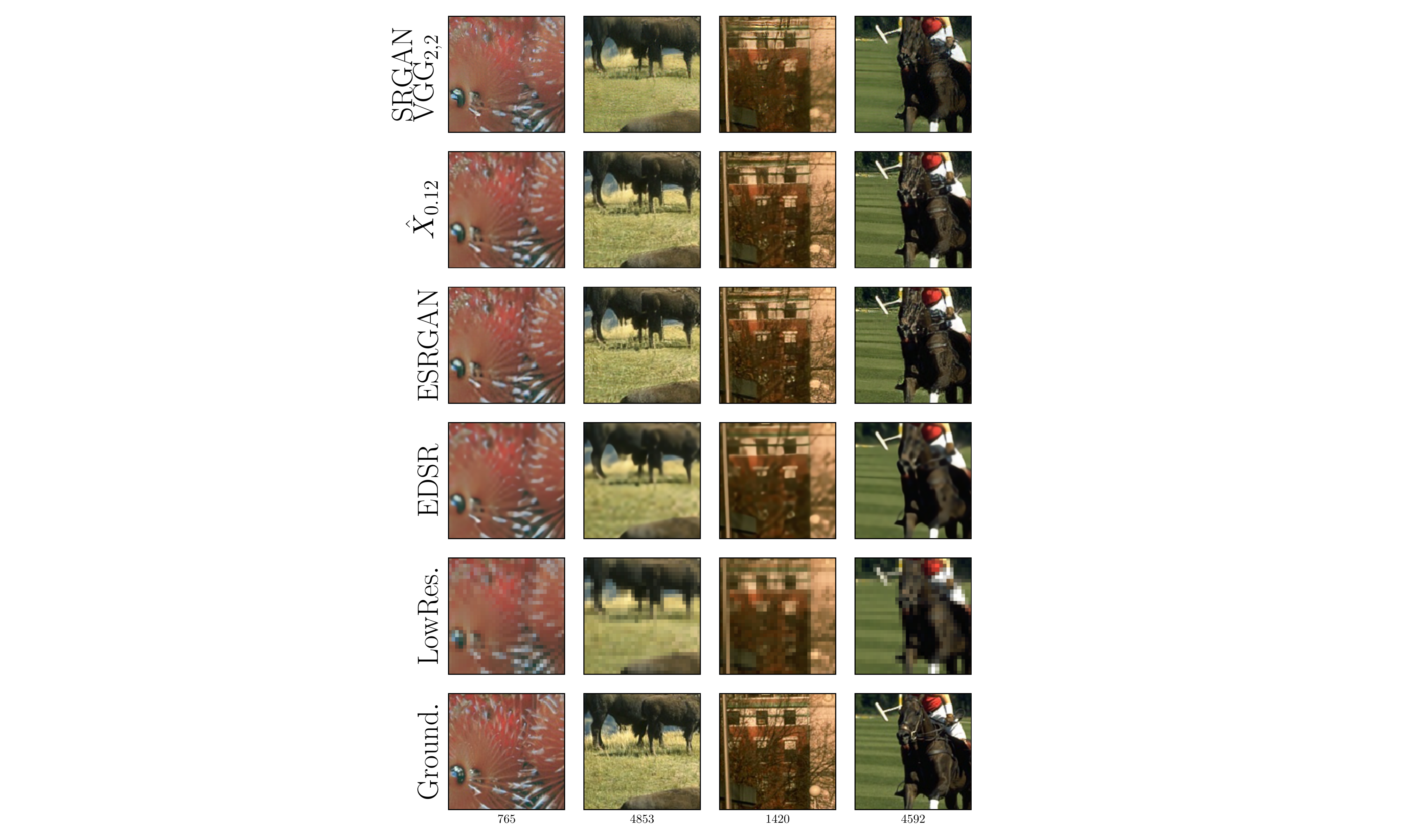}
			\par\end{centering}
	\centering{}\caption{\label{appfig:SrganVS12_full} A visual comparison between  SRGAN-VGG$_{2,2}$ (RMSE: $18.09$, P: $5.03$), and $\hat X_{0.12}$ ($18.15$, $2.48$). }
	\end{figure}
	
\begin{figure}
		\begin{centering}
			\includegraphics[bb=250bp 23bp 850bp 750bp,clip,width = 0.95\linewidth]{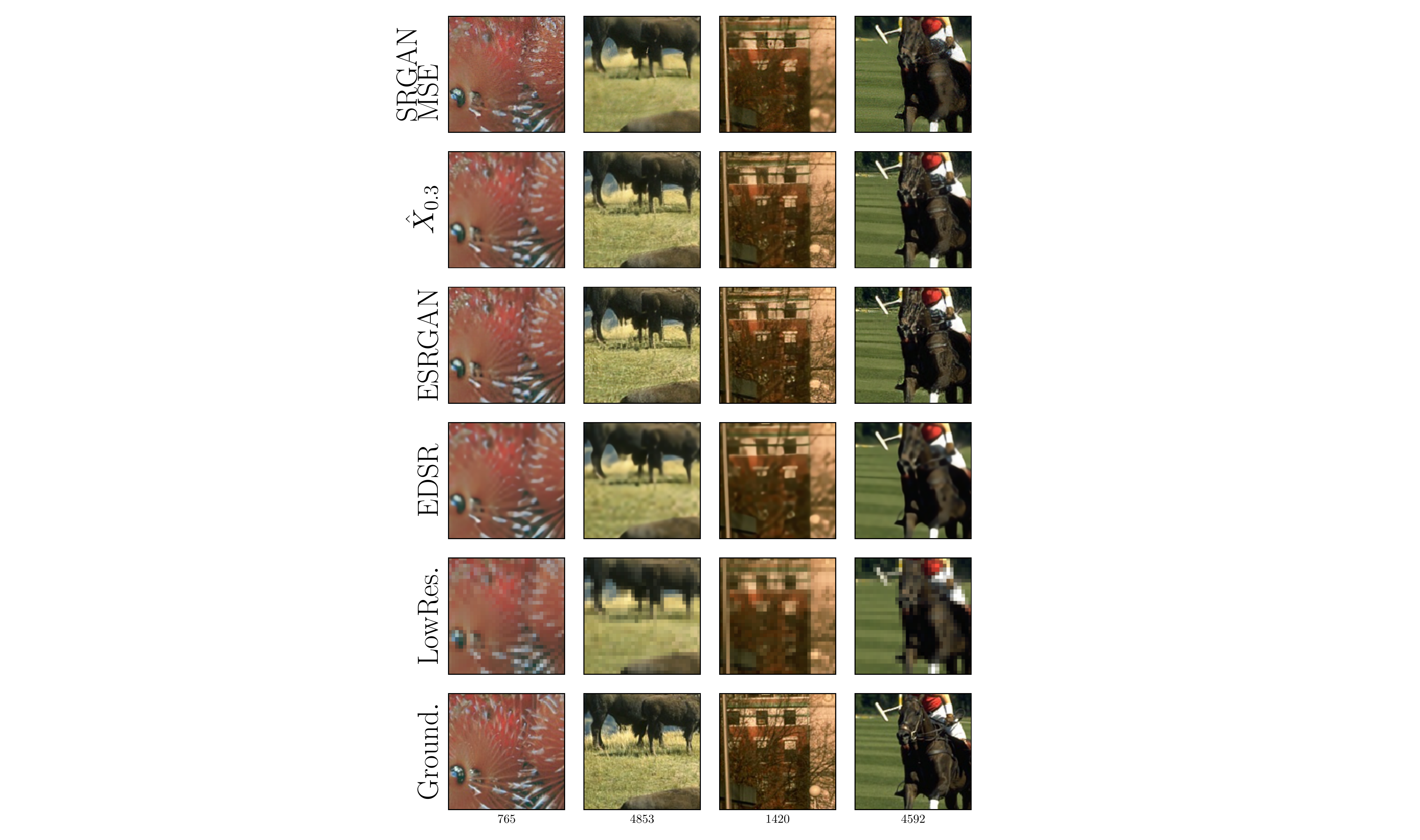}
			\par\end{centering}
	\centering{}\caption{\label{appfig:SrganVS31_full} A visual comparison between  SRGAN-MSE (RMSE: $16.94$, P: $5.88$), and $\hat X_{0.3}$ ($16.82$, $4.15$). }
	\end{figure}


\end{document}